\PassOptionsToPackage{dvipsnames}{xcolor}
\documentclass[conference]{IEEEtran}

\newif\ifnotes

\usepackage[T1]{fontenc}
\usepackage[utf8]{inputenc}
\usepackage{microtype}
\usepackage{graphicx}
\usepackage{mathtools}
\usepackage{booktabs}
\usepackage{multirow}
\usepackage{tabularx}
\usepackage{enumitem}
\usepackage{xspace}
\usepackage{pifont}
\usepackage{relsize}
\usepackage{subcaption}
\usepackage{siunitx} 
\usepackage{listings}
\usepackage{tcolorbox}
\usepackage{flushend}
\usepackage{mdframed}
\usepackage{comment}
\usepackage{wrapfig}
\usepackage{enumitem}

\usepackage{tikz}
\usetikzlibrary{positioning,arrows.meta,calc}

\usepackage[colorlinks,
citecolor=violet,
urlcolor=black,
linkcolor=violet,
bookmarks=false,
hypertexnames=true]{hyperref}

\usepackage[obeyFinal]{todonotes}
\ifnotes
  \usepackage[nomargin,inline,draft]{fixme}
\else
  \usepackage[nomargin,inline,final]{fixme}
\fi
\fxusetheme{color}
\fxuseenvlayout{color}

\newcommand{\jd}[1]{\ifnotes{\textbf{\textcolor{blue}{JD: #1}}}\fi}
\newcommand{\rc}[1]{\ifnotes{\textbf{\textcolor{red}{RC: #1}}}\fi}

\usepackage[plain]{algorithm2e}
\usepackage{algpseudocode}
\usepackage{algcompatible}
\newcommand{\comalgo}[1]{\textcolor{Maroon}{/\!/\,#1}}

\usepackage{amsthm}
\usepackage{amssymb}
\newtheorem{definition}{Definition}
\newtheorem{theorem}{Theorem}
\newtheorem{lemma}{Lemma}

\newcommand{\nameTitle}{Beware}
\newcommand{\name}{\nameTitle}

\IEEEoverridecommandlockouts
\begin{document}



\title{Robust and Automated Reconfiguration of Byzantine Wide-Area Replication}



\author{
    \IEEEauthorblockN{Rowdy Chotkan, Bulat Nasrulin, Johan Pouwelse and Jérémie Decouchant}
   \IEEEauthorblockA{
     Delft University of Technology, The Netherlands\\
     \{r.m.chotkan-1, b.nasrulin, j.a.pouwelse, j.decouchant\}@tudelft.nl
   }
   \thanks{This work was partially funded by NWO/TKI grant BLOCK.2019.004.}
}

\maketitle

\begin{abstract}

Distributed systems handle adversarial nodes through redundancy, which imposes a significant performance overhead. In blockchain systems, Byzantine fault-tolerant state-machine replication (BFT-SMR) is the replicated service that totally orders client transactions before execution. While prior research has primarily focused on designing novel consensus algorithms with improved performance, recent studies have shown that further gains can be achieved through configuration optimization. More precisely, replicas can monitor network latency to dynamically assign the leader role and tune voting weights, thereby improving consensus performance. However, we identify three vulnerabilities in this process that Byzantine nodes can exploit. To address these weaknesses, we propose \name{}, a reconfiguration framework that filters out falsified latency reports, computes robust weight distributions, and applies machine learning to converge towards Byzantine-resilient configurations. Our evaluation shows that \name{} reduces consensus latency by up to 45\% compared to existing solutions. 
\end{abstract}





\section{Introduction}




Blockchains rely on globally distributed replicas to reach consensus. While this distribution enhances resilience and decentralization, it incurs higher, variable network latency, which directly affects the performance of consensus protocols~\cite{de2022noise,uta2018performance,iosup2011performance}.
Crucially, unlike in local networks, latencies in Wide Area Networks (WANs) are highly variable over time, especially across cloud regions or public network infrastructure. As illustrated in Fig.~\ref{fig:wonder_latency}, WAN link latencies can fluctuate by tens of percent within a single day (approximately 20\% in this example) due to dynamic and unpredictable factors such as routing changes and network congestion. Moreover, these variations are not evenly distributed and tend to affect certain regions disproportionately.
To maintain consistently high performance, WAN systems must therefore support dynamic reconfiguration to adapt in real-time to shifting network conditions.

Several mechanisms have been proposed to allow blockchain consensus systems to improve their performance automatically, such as optimistically reducing the system size based on a threat detector~\cite{silva2021threat} or selecting a faster leader~\cite{amir2010prime}. Furthermore, Berger et al.~\cite{berger2020aware,berger2024chasing} demonstrated that the latency of the seminal Practical Byzantine Fault Tolerance (PBFT) algorithm can be reduced in geo-distributed settings by: (i) incorporating additional nodes (i.e., using $3f{+}1{+}\Delta$ replicas instead of $3f{+}1$); (ii) employing two distinct voting weights ($V_\text{min}$ and $V_\text{max}$); and (iii) applying optimization techniques to assign weights and roles (leader or backup) to replicas. This optimization allows for smaller quorums consisting of only $2f{+}1$ replicas with a $V_\text{max}$ weight, compared to the larger quorums of $\lceil \frac{n{+}f{+}1}{2} \rceil$ required in uniformly weighted schemes~\cite{bessani2014state}.

\begin{figure}[t]
    \centering
    \includegraphics[width=0.9\linewidth]{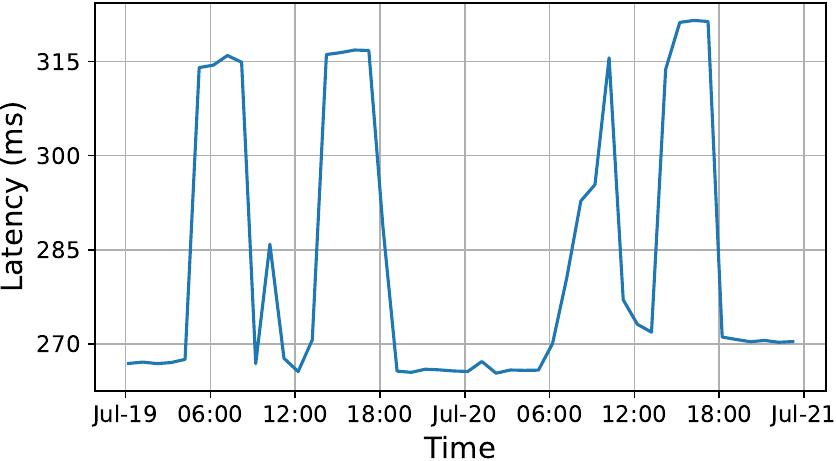}
    \caption{Measured round-trip time between Sydney and Dubai over a 48-hour period (WonderNetwork, July 19--21, 2020).}
    \label{fig:wonder_latency}
\end{figure}

These self-reconfiguration methods rely on replicas themselves to commit their network latencies through the main consensus algorithm, independently compute how the system should be reconfigured, and adopt the new configuration. This approach faces two significant challenges: (i) faulty nodes may report arbitrary latencies, leading the system to be reconfigured into suboptimal configurations; and (ii) faulty nodes may gather high voting weight and later on intentionally delay consensus messages, increasing the overall latency of blockchain consensus. As our experiments demonstrate, these attacks severely compromise reconfiguration effectiveness (see Sec.~\ref{sec:latency_poisoning}). 
Addressing these challenges is crucial to the practical adoption of self-optimization frameworks, ensuring robust, low-latency performance even in adversarial environments.
In practice, this requires deeper integration of system-level telemetry and the ability to distinguish between honest variation and adversarial manipulation.
 
We present \name{} as a solution to these challenges. To filter out inconsistent or malicious latency reports, \name{} uses a robust Virtual Coordinate System (VCS)~\cite{dabek2004vivaldi,seibert2013newton} that sanitizes latency reports. VCSs are traditionally designed for large-scale, permissionless systems in which nodes rely on partial local information. Whereas, in our setting, nodes have access to all latency reports, since they are committed through the main consensus algorithm. \name{} filters out malicious latency reports by leading each replica to locally build a VCS that is iteratively constructed using a clustering method that is unaffected by incorrect latencies.
To limit the impact of Byzantine nodes delaying messages on performance, we extend the weighted voting theory to establish sufficient and efficiently verifiable conditions for a set of weights to define a Byzantine quorum dissemination system (see Theorem~\ref{thm:threshold}, proven in Sec.~\ref{sec:theory}) and efficiently support an arbitrary number of possible voting weights (instead of only two in the literature~\cite{sousa2015separating,berger2020aware}). We demonstrate that supporting a more diverse set of weights enables a more graceful performance degradation when faulty nodes delay consensus messages. Lastly, to adopt a system configuration that maintains low latency despite Byzantine nodes, we employ a machine-learning prediction model. This model learns from past reconfigurations to predict the latency of a given configuration, ensuring optimal performance and faster convergence under adversarial conditions. 

As a summary, we make the following contributions:

\begin{itemize}[leftmargin=*]
\item We identify three key limitations in existing self-optimizing BFT-SMRs that degrade the efficacy of current reconfiguration algorithms. First, Byzantine nodes can prevent the identification of an efficient configuration by sharing poisoned latency reports. Second, the limited number of voting weights employed in current weighted voting algorithms increases their sensitivity to Byzantine faults. Finally, Byzantine nodes can slow down a consensus algorithm by deviating from it, a phenomenon that network latencies alone cannot capture. For each of these three weaknesses, we demonstrate how corresponding Byzantine attacks can degrade consensus performance.

\item We formulate, for the first time, sufficient conditions for a weighted quorum system to be a Byzantine dissemination quorum system (Sec.~\ref{sec:theory}). Verifying whether a candidate quorum system with an arbitrary number of distinct weights satisfies these inequalities is low-complexity (the heaviest operation is sorting the weights). The only comparable method we are aware of is the exhaustive method, which involves computing all quorums, verifying that each quorum is always available, and ensuring that the intersection of any two quorums is large enough to guarantee safety. However, this method suffers from exponential time complexity.

\item We design a BFT-SMR reconfiguration framework, \name{}, that combines several technical innovations. \name{} first leverages a robust, clustering-based virtual coordinate system tailored to fully connected and wide-area adversarial networks. It further extends the weighted voting theory to accommodate a more diverse set of weights, enabling more graceful performance degradation under attack. Finally, \name{} incorporates machine learning to avoid configurations where faulty nodes are likely to increase latency.

\item We evaluate \name{} using realistic simulations and a deployment, in which network conditions are drawn from the WonderNetwork dataset~\cite{wondernetwork}. Our results indicate that \name{} outperforms the state of the art in consensus latency, matrix sanitization, and resilience to adversarial networks. Specifically, \name{} improves the average consensus latency of the PBFT algorithm by up to 45\%  compared to the state-of-the-art reconfiguration framework.
\end{itemize}

This paper is organized as follows. 
Sec.~\ref{sec:model} details our system model, and Sec.~\ref{sec:background} provides some necessary background knowledge on self-reconfiguring BFT-SMR systems.
Sec.~\ref{sec:theory} details conditions that can be efficiently checked to verify whether a weighted quorum system is a Byzantine dissemination quorum system, i.e., whether it is always available and consistent.    
Sec.~\ref{sec:system} presents \name{}, our BFT-SMR reconfiguration framework. 
Sec.~\ref{sec:perfeval} evaluates the performance of \name{} and compares it with the state-of-the-art reconfiguration approaches. Sec.~\ref{sec:sota} discusses the related work. Finally, Sec.~\ref{sec:conclusion} concludes this paper. 

\section{System model} 
\label{sec:model}
    
We consider a system $\Pi = \{p_1, \cdots, p_N\}$ of $N=3f{+}1+\Delta$ nodes. The system operates under the assumption of a static adversary controlling up to $f$ Byzantine replicas that may deviate arbitrarily from a specified protocol. The $N$ replicas are tasked with executing a weighted version of the seminal PBFT~\cite{castro1999practical} partially-synchronous BFT-SMR algorithm. We define the configuration of the consensus system as the distribution of special roles to replicas (e.g., leader or backup), their voting weights, and protocol parameters.  

We assume that replicas are interconnected via a wide-area network, where latencies are heterogeneous and may vary over time. In this context, the system's configuration must be periodically re-evaluated (e.g., after a given number of consensus decisions) to assign voting weights to each replica and attribute the leader's role. Overprovisioning the system by using $\Delta$ additional replicas expands the set of possible configurations the system can use to better adapt to network changes. Replicas independently monitor network latencies and commit them through the consensus algorithm. Faulty replicas may report incorrect latencies or disrupt the protocol execution by delaying the transmission of messages they are expected to send. For simplicity, we assume that replicas experience similar latency distributions across all clients and therefore focus on minimizing consensus latency. In Sec.~\ref{sec:sota}, we discuss how one could instead focus on end-to-end latency.

\section{Reconfiguring Weighted BFT Systems}
\label{sec:background}

In this section, we detail our system model, which follows the standard assumptions of existing reconfiguration frameworks. We then describe the state-of-the-art approach to Byzantine reconfiguration, and identify three problems it faces that limit its effectiveness in adversarial settings.

\subsection{State-of-the-Art Consensus Reconfiguration}

    


\begin{figure*}[ht]
    \centering
    \includegraphics[width=\textwidth]{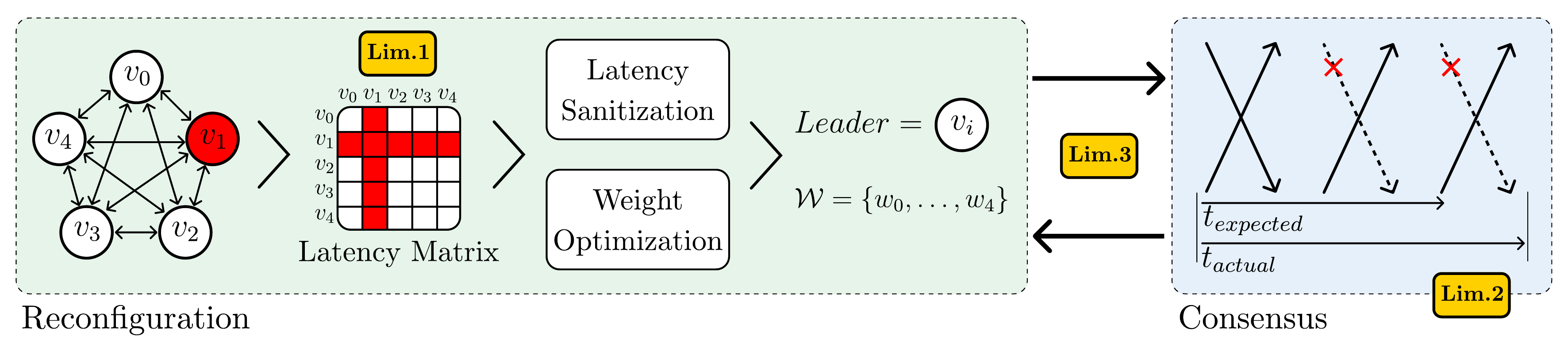}
    \caption{Illustration of limitations 1 to 3, which can undermine a reconfiguration process where replicas commit and sanitize a latency matrix via consensus, optimize voting weights and roles, and adopt the resulting configuration.}
    \label{fig:system_problems}
\end{figure*}

A common approach to improving performance in Byzantine fault-tolerant systems involves dynamic reconfiguration based on observed network conditions~\cite{lamport2010reconfiguring}. 
Following this approach, Aware~\cite{berger2020aware} continuously monitors network conditions and periodically adjusts the configuration of BFT-SMaRt~\cite{bessani2014state}, an implementation of the classical PBFT~\cite{Castro+Liskov:osdi:pbft:1999} algorithm that consists of three communication phases (Propose, Write, and Accept) per request. More specifically, Aware distributes WHEAT~\cite{sousa2015separating}'s voting weights and the leader role among replicas to minimize consensus latency. In a Byzantine system of $n=3f+1+\Delta$ replicas, WHEAT affects a voting weight $v_{max} = (1 + \Delta/f)$ to $2f$ replicas and a voting weight $v_{min} = 1$ to other replicas. Aware aims to assign voting weights and the leader role to better-connected replicas to accelerate quorum formation, with a quorum weight of $2(f+\Delta)+1$.

The reconfiguration process unfolds in five deterministic steps, ensuring all correct replicas reach the same conclusion.

\textbf{Measuring Latencies.} Each correct replica continuously measures the point-to-point network latencies to every other replica in the system. Measurements are one-sided to ensure integrity against faulty replicas. Measurements are taken for critical consensus messages, primarily those of the Propose phase, where a leader broadcasts a block of transactions, and the Write and Accept phases, which contain all-to-all broadcasts used for commitment. 

\textbf{Disseminating and Synchronizing Latency Measurements.} Replicas broadcast their measured latency vectors to all others using the totally ordered underlying BFT-SMaRt consensus protocol.
This ensures that all correct replicas eventually converge on and update their local copies of the global latency matrices, $M^{P}$ (for the Propose phase) and $M^{W}$ (for the Write and Accept phases), with the same, synchronized set of raw measurements.

\textbf{Sanitizing Latency Matrices.}
After a predefined number of consensus instances, all replicas deterministically sanitize their synchronized latency matrices, yielding $\hat{M}^{P}$ and $\hat{M}^{W}$.
This step is crucial to mitigate the influence of malicious (Byzantine) replicas that attempt to submit false or misleading latency reports, thereby ensuring robust decision-making.  Matrices are sanitized by replacing each latency entry with the largest of its two symmetric values (i.e., $M[i,j] = \max(M[i,j], M[j,i]))$. Latency entries that replicas have not committed are set to $+\infty$. 

\textbf{Identifying the Optimal Configuration.} 
Each replica independently and deterministically solves an optimization problem using the sanitized matrices.
Replicas use a latency-prediction function to simulate consensus latency for various candidate configurations (changes in voting-weight distributions and/or leader positions) and rely on simulated annealing to identify the configuration that minimizes consensus latency. 

\textbf{Triggering Reconfiguration.}
If the identified optimal configuration offers a performance improvement exceeding a predetermined threshold (to avoid unnecessary, marginal reconfigurations), all replicas independently obtain the same optimal configuration.
They then collectively trigger a view change, allowing replicas to adopt their new voting weight and role, thereby completing the reconfiguration loop.

 

\subsection{Limitations}

In the following, we describe three weaknesses that can be exploited to degrade the efficacy of current reconfiguration frameworks, illustrated as Lim.~1--3 in Fig.~\ref{fig:system_problems}. \\

\textbf{Limitation 1: Poisoned latency matrix.} Byzantine nodes can communicate incorrect network latencies. While previous approaches rely on simple sanitization methods, such as symmetrizing the latency matrix by replacing each value with the maximum of its mirrored counterpart across the diagonal~\cite{berger2020aware}, we experimentally show that these methods are not immune to Byzantine attacks and do not ensure efficient system reconfiguration. Interestingly, while it was known that pairs of Byzantine nodes can underestimate their latency to appear better connected, we discovered that Byzantine nodes that overestimate their latency to other nodes can also degrade reconfiguration performance. \\

\textbf{Limitation 2: Under-explored Byzantine quorum systems.} In a given system configuration, Byzantine nodes may fail to process incoming messages or refrain from sending messages. As a result, the consensus protocol would achieve lower performance than predicted by the reconfiguration framework. Unlike the emission of conflicting signed messages~\cite{civit2021polygraph}, this type of Byzantine fault cannot be identified under the fluctuating network conditions we consider. Because current weighted voting systems consider only at most two possible voting weights, the impact of these deviations is severe. \\

\textbf{Limitation 3: Ineffective reconfigurations.} A reconfiguration framework deterministically computes the configuration that should deliver the best performance, assuming that all nodes adhere faithfully to the consensus algorithm. However, when this assumption is not met, existing reconfiguration methods repeatedly identify the same underperforming configuration and remain in it. It is therefore necessary for reconfiguration protocols to measure a configuration's actual performance and learn from it to affect future reconfigurations.

\section{Generalizing Weighted Dissemination Quorum Systems}
\label{sec:theory}


Before detailing \name{}, our solution to the aforementioned limitations, we first present some theoretical results that enable \name{} to more effectively explore the space of possible voting weights, thereby increasing the number of system configurations it considers.  
To the best of our knowledge, we are the first to explicitly define formulae for the availability and consistency of Byzantine weighted quorum systems and to formally define the weight requirements of a Byzantine weighted quorum. These theoretical developments enable us to design an optimization algorithm that generates potential voting weights for the nodes in a given consensus system and efficiently verifies whether these weights satisfy the properties of a Byzantine quorum.  
 
This section recalls the definition of a Byzantine weighted dissemination quorum system. It then establishes a sufficient condition for a set of weights to constitute such a quorum system. This condition can be verified in linear time, enabling \name{} to efficiently explore the space of possible replica weight assignments in order to minimize consensus latency.  

\subsection{Dissemination quorum systems}

BFT-SMR systems rely on the concept of Byzantine quorums to reach decisions. In particular, they use dissemination quorums to ensure the self-verification of information, such as signed user transactions. Our objective is to identify Byzantine weighted dissemination quorums and latency-optimal weights for system reconfiguration. Definition~\ref{def:bqs} provides the formal definition of dissemination quorums.

\begin{definition}
A quorum system $\mathcal{Q}$ is a dissemination quorum system~\cite{malkhi1998byzantine} for a fail-prone system $\mathcal{B}$ if the following properties are satisfied.\\
$\bullet$ \textbf{Availability:} $\forall B \in \mathcal{B}, \exists Q \in \mathcal{Q} : B \cap Q = \emptyset$ \\
$\bullet$ \textbf{Consistency:} $\forall Q_1, Q_2 \in \mathcal{Q}, \forall B \in \mathcal{B} : Q_1 \cap Q_2 \nsubseteq B$
\label{def:bqs}
\end{definition}

Instantiating the fail-prone system $\mathcal{B}$ to capture all possible Byzantine sets of size $f$, availability therefore implies that for any possible set of faulty replicas, there exists at least one quorum consisting solely of correct replicas. Consistency guarantees that any two quorums intersect in at least one correct replica.

\subsection{A Sufficient Condition for Weighted Dissemination Quorum Systems}

In the context of weighted voting, Byzantine quorums are defined by their total weight, which must exceed a weight threshold. While previous works~\cite{sousa2015separating} relied on ad-hoc weight assignments, we generalize this concept by formally defining the requirements of weighted dissemination quorum systems.

We assume that each process $P_i$ is assigned a voting weight $w_i \in \{w_1, \cdots, w_N\}$ such that $w_1 \ge \cdots \ge w_{N}$. 
We consider a quorum system $\mathcal{W}$, and note $Q_T$ a weight threshold that is used to define $\mathcal{W}$, i.e., $\mathcal{W}$ is the set of all subsets of $\Pi$ whose total weight is larger than or equal to $Q_T$. 
We note $w(Q)$ the combined weight of the replicas in a set $Q \in \mathcal{W}$.
In the following, we establish sufficient conditions over $Q_T$ for $\mathcal{W}$ to be a dissemination quorum system. 

\begin{lemma} \label{lemma:available}
    $\mathcal{W}$ is available iff. $Q_T {\le} \sum_{i \in [1,N]} w_i {-} \sum_{i \in [1,f]} w_i$.
\end{lemma}

\begin{proof}

\textbf{$(\Rightarrow)$} 
Let us assume that $\mathcal{W}$ is available, and consider the set $Q$ of the $N-f$ nodes with the lowest weights. 
The value of $w(Q)$ is $\sum_{i \in [1,N]} w_i - \sum_{i \in [1,f]} w_i$.
By assumption, $Q \in \mathcal{W}$, because the system needs to be available if the $f$ nodes with the highest weights are faulty, which means that $Q_T \le w(Q)$. 
We then obtain $Q_T {\le} \sum_{i \in [1,N]} w_i {-} \sum_{i \in [1,f]} w_i$.

\textbf{$(\Leftarrow)$} Let us assume that $Q_T \le \sum_{i \in [1,N]} w_i - \sum_{i \in [1,f]} w_i$. 
Let us consider a set $B$ of $f$ faulty nodes. 
The total weight $w(B)$ of this set is lower than or equal to the sum of the $f$ largest weights $\sum_{i \in [1,f]} w_i$. 
The combined weight of the set $\Pi \setminus B$ of the $N-f$ correct nodes, $\sum_{i \in [1,N]} w_i - w(B)$, is then larger than or equal to $\sum_{i \in [1,N]} w_i - \sum_{i \in [1,f]} w_i$. By assumption, we then obtain that $\Pi \setminus B \ge Q_T$, which means that $\Pi \setminus B \in \mathcal{W}$ and that $\mathcal{W}$ is available.  
\end{proof}

\begin{lemma} \label{lemma:safe}
    If $\left( \sum_{i \in [1,N]} w_i {+} \sum_{i \in [1, f]} w_i \right) {/} 2 < Q_T$ then $\mathcal{W}$ is consistent. 
\end{lemma}

\begin{proof}


    By contradiction, let us assume that the implication is false, i.e., that $( \sum_{i \in [1,N]} w_i {+} \sum_{i \in [1, f]} w_i ) {/} 2 < Q_T$ and that $\mathcal{W}$ is not consistent.
    Let us consider two sets $Q_1, Q_2 \in \mathcal{W}$ such that $|Q_1 \cap Q_2|$ contains $f$ or less replicas. These sets are guaranteed to exist because $\mathcal{W}$ is not consistent.
    The union of $Q_1$ and $Q_2$ contains at most all the replicas, i.e., 
    $w(Q_1 \cup Q_2) \le \sum_{i \in [1, N]} w_i$, which can be rewritten as:
    \begin{equation} \label{eq:inequation}
    w(Q_1) + w(Q_2) \le w(Q_1 \cap Q_2) + \sum_{i \in [1, N]} w_i  
    \end{equation}

    Let us now make two observations. First, the total weight of $Q_1 \cap Q_2$ is lower than or equal to the combined weight of the $f$ heaviest replicas, because $Q_1 \cap Q_2$ contains $f$ or fewer replicas. This can be written as $w(Q_1 \cap Q_2) \le \sum_{i \in [1, f]} w_i$.
    Second, $Q_T \le w(Q_1)$ and $Q_T \le w(Q_2)$ because $Q_1, Q_2 \in \mathcal{W}$.
    Using those two observations in Inequation~\ref{eq:inequation}, we obtain
    $2 Q_T \le \sum_{i \in [1, f]} w_i + \sum_{i \in [1, N]} w_i$, which contradicts our assumption.    
\end{proof}



\begin{theorem} \label{thm:threshold}
Let $\Pi = \{p_1, \cdots, p_{N}\}$ be a system where each node $p_i$ has a weight $w_i$, with $w_1 \ge w_2 \ge \cdots \ge w_N$. 
Let us consider a weight threshold $Q_T$ (if any) that verifies 
\[
\left( \sum_{i \in [1,N]} w_i + \sum_{i \in [1, f]} w_i \right) / 2 < Q_T \le \sum_{i \in [1,N]} w_i - \sum_{i \in [1,f]} w_i.
\]
The set of all subsets of $\Pi$ whose total weight $Q_T$ is 
larger than or equal to $Q_T$ is a Byzantine dissemination quorum system.
\end{theorem}

\begin{proof}
    Immediate using Definition~\ref{def:bqs}, and Lemmas~\ref{lemma:available} and~\ref{lemma:safe}.
\end{proof}


\noindent
Note that in the egalitarian case, i.e., when every node has a unitary voting weight, Theorem~\ref{thm:threshold} implies that a Byzantine quorum requires at least $\lceil \frac{n{+}f{+}1}{2} \rceil$ votes, which is a known formula. Similarly, given Aware's~\cite{berger2020aware} weight distribution, Theorem~\ref{thm:threshold} provides $2f{+} 2\delta{+}1/2$ as a strict lower bound for the weight threshold, which confirms the $2f{+}2\delta{+}1$ weight threshold of Aware.
We leverage Theorem~\ref{thm:threshold} to efficiently verify whether a set of weights can form a Byzantine-weighted dissemination quorum system. Note that Theorem~\ref{thm:threshold} might not characterize all possible Byzantine dissemination quorum systems because Lemma~\ref{lemma:safe} is not an equivalence. \\

\begin{figure*}
    \centering
    \includegraphics[width=0.9\linewidth]{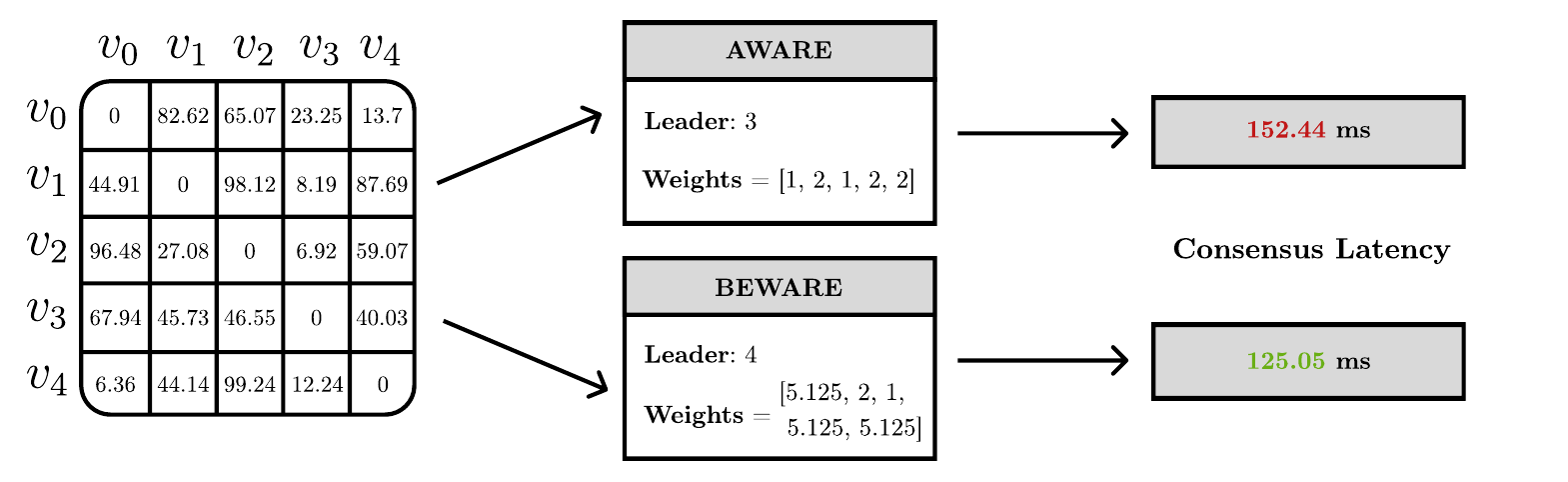}
    \caption{Given identical network latencies, Beware's leader and weight selection yields approx. $125$\,ms consensus latency, outperforming Aware's approx. $152$\,ms.}
    \label{fig:consensus_example}
\end{figure*}

\textbf{Example.} Fig.~\ref{fig:consensus_example} showcases an example in which considering more than two voting weights, and doing so in more flexible ways, following Theorem~\ref{thm:threshold}, enables new configurations that further improve performance.
In this example, we consider $n=5$ replicas with $f=\Delta=1$. All replicas are correct and collectively report the indicated latency matrix. In the best configuration that Aware identifies, replica $3$ is the leader, and the replicas are assigned weights $[2,1,1,2,2]$, using $5$ as the weight threshold to obtain a configuration where PBFT's consensus latency equals approx. $152$\,ms. On the other hand, Beware is able to identify a better configuration where replica $4$ would act as the leader and replicas would be given weights $[5.125, 2, 1, 5.125, 5.125]$, and use a weight threshold $Q_T$ such that $11.75 < Q_T \le 13.25$ (e.g., 12), which would decrease the consensus latency down to approx. $125$\,ms (an 18\% improvement over Aware). Note that a seemingly close configuration that Aware could have identified would select replica $4$ as the leader and choose weights $[2, 1, 1, 2, 2]$. However, the latency of this configuration would be approx. $160$\,ms. While this example only considers correct nodes, the gap between Aware and Beware is sometimes even larger (up to $50$\%) when faulty replicas are present (cf. Fig.~\ref{fig:fluctuations}).

\section{\name{}}
\label{sec:system}

    \name{} follows the approach pioneered by state-of-the-art reconfiguration frameworks while addressing their limitations.
    To do so, \name{} combines three novel methods. Succinctly, \name{} obtains a reliable view of network latencies in wide-area networks through a novel Byzantine-resilient virtual coordinate system, explores a more diverse set of voting weights to enhance performance under diverse conditions, and achieves accurate performance predictions of a given configuration based on historical configurations and their observed outcomes. 
    

    \jd{have illustration with all components involved}

    \subsection{Sanitizing the Latency Matrix}

    To address \textbf{Limitation~1}, \name{} sanitizes latency matrices using a Byzantine-resilient clustering defense~\cite{nguyen2022flame,fereidooni2023freqfed,cox2025catalyst}, originally designed for federated learning, that we integrate within a virtual coordinate system~\cite{dabek2004vivaldi,seibert2013newton}. 
    
    Unlike classical VCS algorithms---decentralized and typically tailored to peer-to-peer systems---our approach is centralized but replicated. Given that all nodes that decide on a new configuration eventually need to obtain a consistent latency matrix, the sanitization process can either be executed by all nodes or delegated to $2f{+}1$ nodes, which would then disseminate the sanitized matrix to all nodes. For ease of presentation, we focus here on the former option. 
    Algorithm~\ref{alg:our_vcs} details the pseudocode of our VCS-based sanitization method, which takes as input the latency matrix $RTT$, where ${RTT}_{i,j}$ represents the (possibly incorrect) latency reported by node $i$ for its connection with node $j$.
    Initially, every node $j$ is assigned random coordinates in the virtual space and an error $e_j$, which expresses the system's confidence in the accuracy of the node's coordinates.
    This algorithm performs $T$ iterations, adjusting the position of node $i$ based on the forces exerted by the other nodes.
    The force exerted by a node $j$ is computed based on the difference between ${RTT}_{i,j}$, the reported latency between nodes $i$ and $j$, and the distance between nodes $i$ and $j$ in the virtual space. It is directed along the unitary vector $\overrightarrow{u}_{i,j}$ from nodes $i$ to $j$ (Alg.~\ref{alg:our_vcs}, l.~\ref{line:cluster_force}). Additionally, this force is scaled based on the coordinate errors $e_i$ and $e_j$ (Alg.~\ref{alg:our_vcs}, l.~\ref{line:cluster_coordinate_errors}), such that greater uncertainty in the coordinates of nodes $i$ and $j$ results in a reduced norm of the force. Note that if both $e_i$ and $e_j$ are $0$, the computation would involve division by zero. In this edge case, node $j$'s influence is discarded, as it exerts no force on node $i$. For simplicity, these steps are omitted in the pseudocode. Two parameters, $c_e$ and $c_c$, control the extent to which the coordinate error and node coordinates are updated during each iteration.

    To filter out harmful Byzantine forces, Alg.~\ref{alg:our_vcs} groups the forces that all nodes exert on node $i$ into clusters using \textsc{HDBSCAN}~\cite{campello2013density}, based on their cosine distance (Alg.~\ref{alg:our_vcs}, l.~\ref{line:cluster_clustering}).
    Given up to $f$ Byzantine nodes in a system of at least $3f{+}1$ nodes, the cluster containing the majority of forces is dominated by correct nodes, while any remaining Byzantine influence that belongs in this cluster has little impact. 
     Forces in the largest cluster are clipped to their median norm and then averaged to update node $i$ 's coordinate (Alg.~\ref{alg:our_vcs}, l.~\ref{line:cluster_update_coordinate}). 

     In our implementation, we use a $3$-dimensional space, with $e_j=0.1$, $c_e=c_c=0.25$, and $T=1,000$, based on empirical tuning. These values strike a good balance between responsiveness and stability during convergence, yielding stable and accurate virtual coordinates across our test scenarios.
     In a system of $N=100$ replicas, Algorithm~\ref{alg:our_vcs} runs in approximately $160$\,s for $T=1,000$ iterations. Note that, based on Fig.~\ref{fig:wonder_latency}, this is sufficient to address real-life latency variations. In addition, the reconfiguration process is executed periodically and asynchronously at epoch boundaries and does not block ongoing consensus decisions. In our current implementation of \name{}, this code is currently written in Python and can therefore be optimized.

    \RestyleAlgo{boxruled}
    \LinesNumbered
    \begin{algorithm}[t]
    \begin{algorithmic}[1]
    \State \textbf{Input:} Latency matrix ${RTT}$ of size $N \times N$
    \State \textbf{Output:} Sanitized latency matrix
    \State \textbf{Parameters:} $0 < c_c, c_e < 1$ 
    \item[]
    
    \State \textbf{for} $iter$ \textbf{in} $range(T)$:  
        \State \hspace*{1em} \textbf{for} $i$ \textbf{in} $[1, N]$:
            \State \hspace*{2em} \textbf{for} $j$ \textbf{in} $[1, N] \setminus \{i\}$: \comalgo{forces exerted on $i$}
                \State \hspace*{3em} $w[j] = e_i / (e_i + e_j)$ \label{line:cluster_coordinate_errors}
                \State \hspace*{3em} $e_s[j] = \left| ||x_i - x_j|| - {RTT}_{ij} \right| / {RTT}_{ij}$
                \State \hspace*{3em} $\overrightarrow{F}[j] = w[j] \cdot ({RTT}_{ij} - ||x_i - x_j||) \cdot \overrightarrow{u}_{i,j}$ \label{line:cluster_force}
        
        \State 
        \State \hspace*{2em} $(b_1, {\cdots}, b_L) = \textsc{Cluster}(\{\overrightarrow{F}_{j \neq i}[j]\})$ \label{line:cluster_clustering} \comalgo{Largest cluster}
        \SetAlgoCaptionSeparator{}
        \State \hspace*{2em} $\overline{e_s} = \textsc{Mean}(e_s[b_1], \cdots, e_s[b_L])$
        \State \hspace*{2em} $\overline{w} = \textsc{Mean}(w[b_1], {\cdots}, w[b_L])$
        \State \hspace*{2em} $e_i = c_e {\cdot} 
        \frac{\sum_{k \in [1, L]}{w[b_k] \cdot e_s[b_k]}}{\sum_{k \in [1, L]}{w[b_k]}} + (1 - c_e) {\cdot} e_i$
        \State
        \State \hspace*{2em} $\overline{|F|} = \textsc{Median}(|\overrightarrow{F}[b_1]|, \cdots, |\overrightarrow{F}[b_L]|)$
        \State \hspace*{2em} $x_i = x_i + c_c \cdot \textsc{Mean}\left(\overrightarrow{F}[b_1], \cdots, \overrightarrow{F}[b_L]\right) / \overline{|F|}$ \label{line:cluster_update_coordinate}
    
    \State
    \State \textbf{return}\,$\{ || x_i {-} x_j ||,\,{\forall}1\,{\le} i,j {\le} N \}$ \comalgo{Sanitized latency matrix}
    \end{algorithmic}
    \vspace{2mm}
    \caption{Robust centralized VCS construction using clustering and clipping of forces.} 
    \label{alg:our_vcs}
    \end{algorithm}

\subsection
{Optimizing Byzantine Weighted Dissemination Quorums for Latency}  

To find latency‐efficient weighted quorum assignments and address \textbf{Limitation~2}, \name{}  can explore the joint space of possible leader choices and node weights using one of two optimization approaches: Simulated Annealing (SA) and Differential Evolution (DE). \jd{looks like there are 2 solutions, keep the best}
As a consequence, \name{} supports an arbitrary number of distinct voting weights; each candidate configuration explored during optimization must be verified as a valid Byzantine dissemination quorum system. Naively checking whether a candidate set of weights defines a Byzantine quorum system requires enumerating all possible quorums and verifying safety and liveness (see Def.~\ref{def:bqs}), which is exponential in complexity. Instead, we derive sufficient conditions that can be verified in linear time and that are summarized in Theorem~\ref{thm:threshold}. \\




\textbf{Simulated Annealing}
extends the approach introduced by Berger et al.~\cite{berger2020aware} to arbitrary weights: starting from the current configuration (a leader and a weight vector), it iteratively mutates the configuration to reduce latency while gradually \textit{cooling down}. Concretely, at each SA iteration, a neighboring configuration is generated by either mutating the voting weights or by choosing a different leader. Weight mutations can involve: (i) small incremental changes ($\pm 1$ to individual weights), (ii) random replacement of selected weights with values from the allowed range $[1,10]$, or (iii) heuristic-based adjustments that guide weights toward configurations more likely to satisfy quorum constraints. The newly proposed configuration is accepted if it yields lower latency, or---if it increases latency---with probability $\exp\bigl(-\Delta_{\text{latency}}/T\bigr)$, where $\Delta_{\text{latency}}$ is the latency increase and $T$ is a temperature parameter updated as $T \leftarrow T \cdot (1 - \theta)$. Early in the run, when $T$ is high, SA can escape local minima by occasionally accepting worse configurations. As $T$ cools down, it becomes increasingly unlikely to accept inferior solutions. While SA has the advantage of requiring only a few hyperparameters, its local mutation moves and fixed annealing schedule tend to converge slowly in high-dimensional weight spaces and often get trapped once the temperature is low.

\textbf{Differential Evolution}
is a population‐based evolutionary algorithm~\cite{storn1997differential}. Each candidate in DE's population of configurations is an $(n+1)$-dimensional vector comprising $N$ continuous weight values and a single integer representing the leader index. We maintain a population of $20$ such candidates. At each generation, for every target vector $\mathbf{x}_t$ in the population, DE selects three distinct individuals $(\mathbf{x}_a,\mathbf{x}_b,\mathbf{x}_c)$ uniformly at random and forms a mutant vector
\(
  \mathbf{m} \;=\; \mathbf{x}_a \;+\; F \,\bigl(\mathbf{x}_b - \mathbf{x}_c\bigr)
\), 
where the scaling factor $F$ is typically set to $0.5$. To ensure that every coordinate stays within the valid range (i.e., weights in $[0.25,2.5]$ and leaders in $\{0,1,\dots,n-1\}$), we clip the mutant accordingly. A trial vector is then produced by performing crossover between $\mathbf{m}$ and $\mathbf{x}_t$. For each dimension, we inherit from $\mathbf{m}$ with probability $C_r=0.7$ or from $\mathbf{x}_t$ otherwise. After rounding the leader coordinate to the nearest valid integer, we evaluate the trial's consensus latency and compare it with that of $\mathbf{x}_t$. If the trial has strictly lower latency, it replaces $\mathbf{x}_t$ in the population; otherwise, $\mathbf{x}_t$ survives for the next generation. 

Crucially, we bias the initial DE population so that weight vectors are drawn from a truncated Beta$(5,2)$ distribution on the interval $[0.25,2.5]$ rather than from a uniform distribution. This heuristic avoids extreme all‐ones or all‐twos vectors at the outset, thereby accelerating convergence toward promising regions of the search space while preserving sufficient diversity to explore substantially different configurations. Because DE evaluates $20$ candidates in parallel each generation and relies on differential combinations of existing solutions, it naturally explores the high-dimensional space far more efficiently than SA's single-chain, random-swap approach. Empirically, DE converges to near-optimal weight and leader configurations in far fewer consensus-round evaluations; it is also far less prone to getting stuck in local minima. For these reasons, DE is chosen as the default optimizer in all of our experiments, with SA retained purely as a baseline for direct comparison. Whenever we report ``\name{} (DE)'' in our evaluation, it refers to the version that uses DE with Beta$(5,2)$ initialization; ``\name{} (SA)'' refers to the SA baseline following Berger et al.~\cite{berger2020aware}.

While global optimization techniques identify latency-efficient Byzantine dissemination quorums, they inherently involve predicting the latency of candidate configurations in adversarial settings. 
For this purpose, \name{} integrates predictive modeling into the optimization loop.

\subsection{Predicting Latency and Converging towards the Best Configuration}

\name{} attaches a lightweight predictor to its optimization loop to avoid getting stuck in suboptimal configurations (\textbf{Limitation~3}) or having to temporarily adopt and monitor each candidate configuration, which would be inefficient. Each replica keeps a fixed-size sliding window of the last $W$ reconfigurations. Every entry stores the leader identifier, the $N$ voting weights, and the actual latency observed once the configuration became active. When a new measurement arrives, the oldest record is discarded, ensuring that the training set always reflects current network conditions and does not overfit to stale patterns.

The predictor is trained and periodically retrained in two phases.
During a short \emph{bootstrapping} phase, the optimizer explores normally
and logs $\langle\text{configuration},\text{latency}\rangle$
pairs; after $S=10$ examples, the first model is fit using
XGBoost~\cite{xgboost}.  
Subsequently, whenever the optimizer proposes a candidate
$\mathbf{x}$, it first asks the model for a latency estimate
$\hat{\ell}(\mathbf{x})$.  
An expensive performance monitoring of a candidate configuration is performed only if
\(
\bigl|\hat{\ell}(\mathbf{x})-\ell_{\text{obs}}\bigr|
    >\varepsilon_{\text{desired}},~
\varepsilon_{\text{desired}} = 10^{-2}
\).
If the estimate is accurate, the optimizer can accept or reject
$\mathbf{x}$ immediately; otherwise the fresh ground-truth measurement
$\ell_{\text{obs}}$ is inserted into the window, improving the model.
A \emph{random audit} with probability $p_{\text{check}}=0.1$ per step
forces an explicit measurement even when the error stays just below
$\varepsilon_{\text{desired}}$, ensuring that rare but large drifts are
eventually detected.
The model is retrained whenever the
running relative error
\(
\frac{1}{W}\sum_{i=1}^{W}
\frac{\lvert\hat{\ell}_i-\ell_i\rvert}{\ell_i}
\) 
exceeds $\varepsilon_{\text{retrain}} = 3\times10^{-2}$. 
Retraining is cheap because $W$ is capped, and XGBoost scales linearly in
the number of samples. \\

\textbf{Feature representation.}
For a system with $N$ replicas, we encode a configuration as an
$(N{+}1)$-dimensional real vector
\(
\mathbf{x}= [\,w_1,\dots,w_N,L\,]^{\!\top},
\) 
where $w_i\in[0.25,2.5]$ is the weight of replica~$i$ and
$L\in\{0,\dots,N-1\}$ is the leader index
(normalised to $[0,1]$ before fitting). 
We empirically found that the $[0.25, 2.5]$ weight interval allows \name{} to efficiently explore the weight space and identify better-weighted quorums in a short amount of time. This range is justified by two primary factors: first, it constrains weights to positive values while maintaining a 10:1 ratio between the maximum and minimum bounds; and second, it provides a contained search window, which facilitates the time-efficient exploration of weight configurations.
XGBoost captures the non-linear interactions between weights and the leader
choice; nevertheless, for transparency, we also keep an ordinary
least-squares fit
\(
\hat{\ell}(\mathbf{x})=\beta_0+\sum_{j=1}^{N}\beta_j w_j+\beta_{N+1}L
    +\varepsilon
\),
which runs $\approx 15\%$ less accurately than the boosted model, but
serves as a simple sanity check.
With $W\!=\!200,\,S\!=\!10,\,
\varepsilon_{\text{desired}}\!=\!0.01,\,
\varepsilon_{\text{retrain}}\!=\!0.03,\,
p_{\text{check}}\!=\!0.1$,
the predictor reduces the number of full consensus-round evaluations by
$\sim 60$--$70\%$.

\section{Performance Evaluation}
\label{sec:perfeval}

    We compare \name{} against three baselines. We first consider Aware~\cite{berger2020aware}, a state-of-the-art BFT-SMR reconfiguration framework that sanitizes latency matrices by replacing each value with its symmetric counterpart across the diagonal if it is larger. Then, we consider a PBFT baseline that uses the Newton~\cite{seibert2013newton} robust virtual coordinate system (VCS) to sanitize the latency matrix. 
  We also include a vanilla PBFT reference measured without attacks, denoted as ``No-Attack'' in relevant figures, which serves as an ideal (attack-free) performance baseline. Sections~\ref{sec:latency_poisoning}--\ref{sec:adaptation_network_changes} report results obtained using microbenchmark simulations that each isolate a key component of \name{} and evaluate its performance. More precisely, these sections assess the building blocks of Beware using real-world WAN latency datasets, both with and without faulty behaviors, and with real latency variations. These simulations are executed across 20 nodes on the DAS cluter~\cite{bal2016medium}, equipped with dual 8-core processors
 running at 2.4 GHz, 60\,GB of RAM, and the Rocky Linux operating system.
    Section~\ref{sec:deployment} reports the performance of a full Java implementation of \name{} in a deployment on the same academic cluster. We replay real-world WAN latency traces and show that \name{} performs well in an online setting.
    
    In our experiments, we consider nodes that communicate over network latencies randomly sampled from the WonderNetwork dataset~\cite{wondernetwork}, a comprehensive and up-to-date dataset sourced from a large global networking solution provider. 
    In an experiment, we position each node at a random site in the WonderNetwork dataset and compute the latency matrix based on the latency between sites. All code used for our simulations and deployment experiments is publicly available in an online repository.\footnote{\url{https://data.4tu.nl/datasets/68ee956e-4144-405a-8364-64c233b1555f}}

    \jd{Describe experimental setup here: hardware, network, OS, etc.}

\subsection{Consensus Latency under Latency Poisoning Attacks}
\label{sec:latency_poisoning}

Malicious nodes may report incorrect network latencies to enhance their perceived network latency and degrade consensus latency, thereby influencing the weight assignment process. 
We compare our sanitization method with previous ones using three attacks:
\begin{itemize}
    \item \textbf{Inflation attack:} Byzantine nodes increase all their link latencies.
    \item \textbf{Deflation attack:} Byzantine pairs of nodes decrease their link latencies.
    \item \textbf{Inflation-deflation attack:} Byzantine nodes increase their link latencies with correct nodes and decrease them with other Byzantine nodes. 
\end{itemize}

    

For these experiments, we use Aware's weight-assignment procedure to isolate the effect of latency sanitization mechanisms. For the latency inflation and deflation attacks, we use a factor of 3.
In all experiments, the system size is set to $N=3f{+}1+\Delta$, with $\Delta = f$, following Aware's configuration. 
%
%

\setlength{\intextsep}{5pt} 
\setlength{\textfloatsep}{5pt} 
\setlength{\floatsep}{5pt} 

\begin{table}
\centering
\caption{Latency attack impact measured as the percentage increase over baseline latency}
\label{tab:attack_overview}
\resizebox{\linewidth}{!}{%
\begin{tabular}{|l|c|c|c|}
\toprule
\textbf{Attack} & \textbf{Aware~\cite{berger2020aware}} & \textbf{Newton~\cite{seibert2013newton}} & \textbf{\name{}} (this work) \\
\midrule
Deflation              & 9.58\% & 14.38\% & 3.35\% \\
Inflation              & 9.54\% & 14.63\% & 2.82\% \\
Inflation-def.         & 9.52\% & 13.87\% & 2.82\% \\
\bottomrule
\end{tabular}%
}
\end{table}
\begin{figure}[t]
    \centering
    \begin{subfigure}{0.49\textwidth} 
        \includegraphics[width=\linewidth]{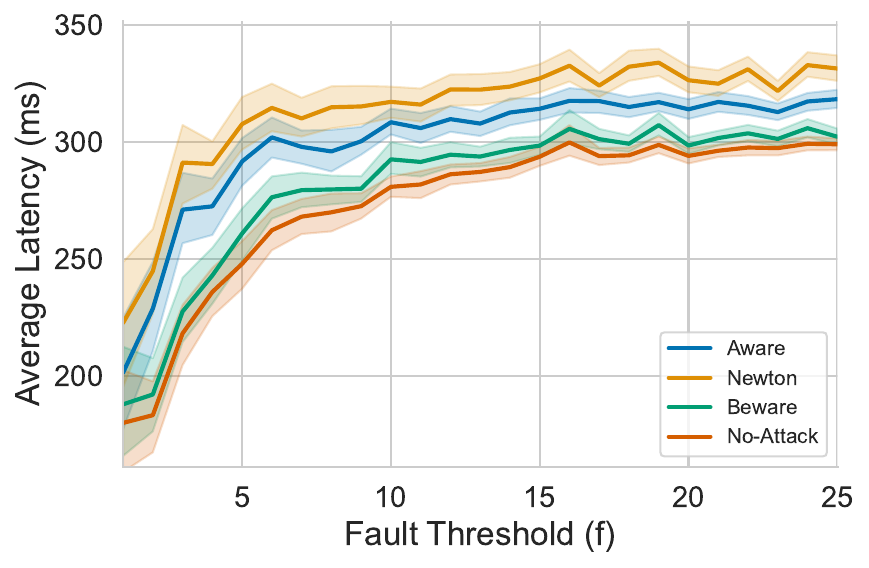}
        \caption{Average consensus latency with $f\in[1,25]$}
        \label{fig:vcs_differing_f}
    \end{subfigure}
    \hfill
    \begin{subfigure}{0.49\textwidth} 
        \includegraphics[width=\textwidth]{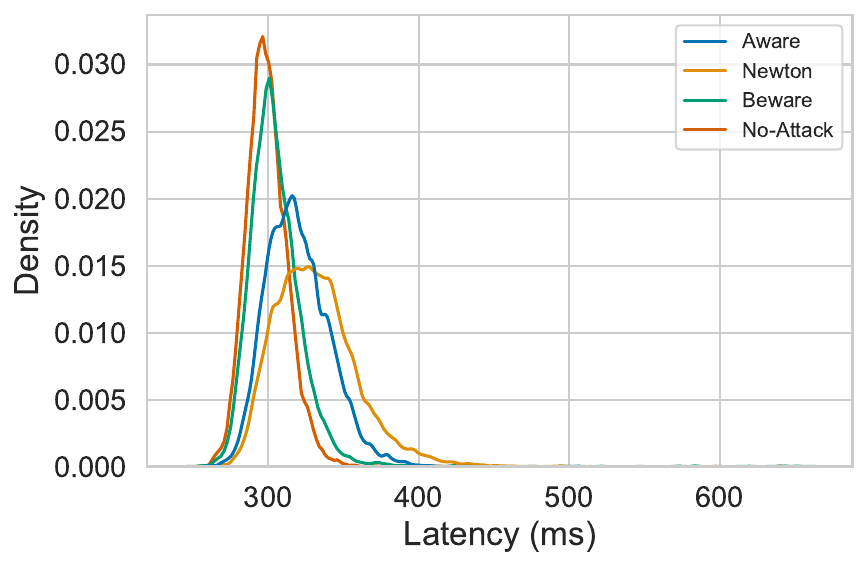}
        \caption{Latency probability density function (PDF) with $f=25$}
        \label{fig:vcs_inf_cds_dens}
    \end{subfigure}
    \caption{Consensus latency under inflation attack.}
\end{figure}
%
Fig.~\ref{fig:vcs_differing_f} shows the average consensus latency under the inflation attack for all algorithms with $f\in[1,25]$, $\Delta=f$, $n=3f+1+\Delta$, and 100 repetitions per setting. The shaded bands show the 95\% confidence interval of the mean across repetitions (computed as $\mu \pm 1.96\cdot \sigma/\sqrt{R}$ with $R=100$).
For all values of $f$ in $[1, 25]$, \name{} provides the most vigorous defense against the inflation attack, with Newton being the most vulnerable baseline and Aware performing in between. This observation also applies to the deflation and inflation-deflation attacks, as shown in Tab.~\ref{tab:attack_overview}. For all values of $f$, \name{} is the most robust defense against the inflation attack, limiting the average latency increase to $8.23$\,ms, compared to $26.68$\,ms for Aware and $38.90$\,ms for Newton. A similar pattern holds for the deflation and inflation-deflation attacks, as shown in Tab.~\ref{tab:attack_overview}.

%
%
Fig.~\ref{fig:vcs_inf_cds_dens} shows the probability density function of the consensus latency for all algorithms under an inflation attack, with $f=\Delta=25$ and $N=3f+1+\Delta=101$. We use $10,000$ repetitions. This figure uses $f=25$ to maximize the number of Byzantine replicas and obtain a stable view of the spread and tails of the latency distribution. \name{}'s distribution is very close to the optimal one, while Aware and Newton's distributions have a higher mean and a larger spread.   


\begin{figure*}[t]
    \centering
    \begin{subfigure}[t]{0.49\textwidth}
        \centering
        \includegraphics[width=0.95\linewidth]{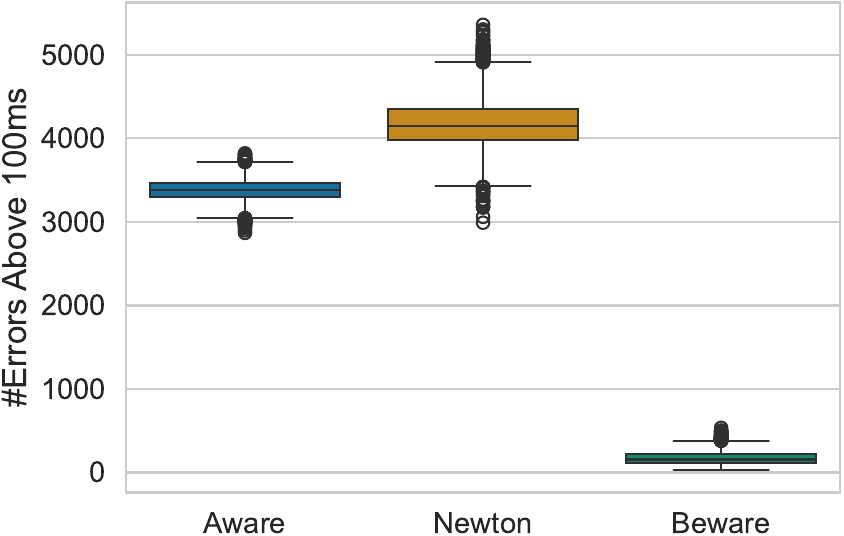}
        \caption{Number of latency matrix entries whose absolute error is above $100$\,ms} \label{fig:vcs_sanitization_a}
    \end{subfigure}
    \hfill
    \begin{subfigure}[t]{0.49\textwidth}
        \centering
        \includegraphics[width=0.95\linewidth]{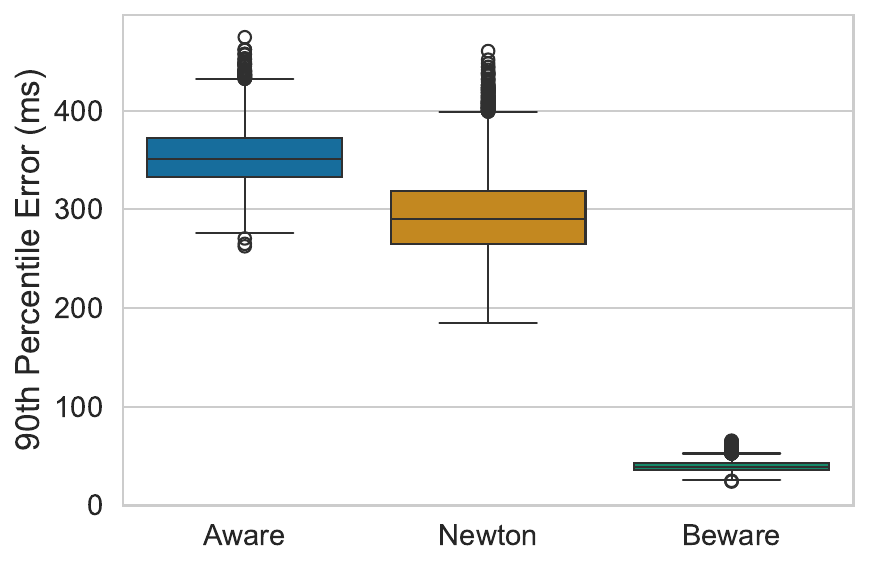}
        \caption{90th-percentile sanitization error}
        \label{fig:vcs_sanitization_b}
    \end{subfigure}
    \hspace*{\fill}

    \vspace{1em}

    \begin{subfigure}[t]{0.49\textwidth}
        \centering
        \includegraphics[width=0.95\linewidth]{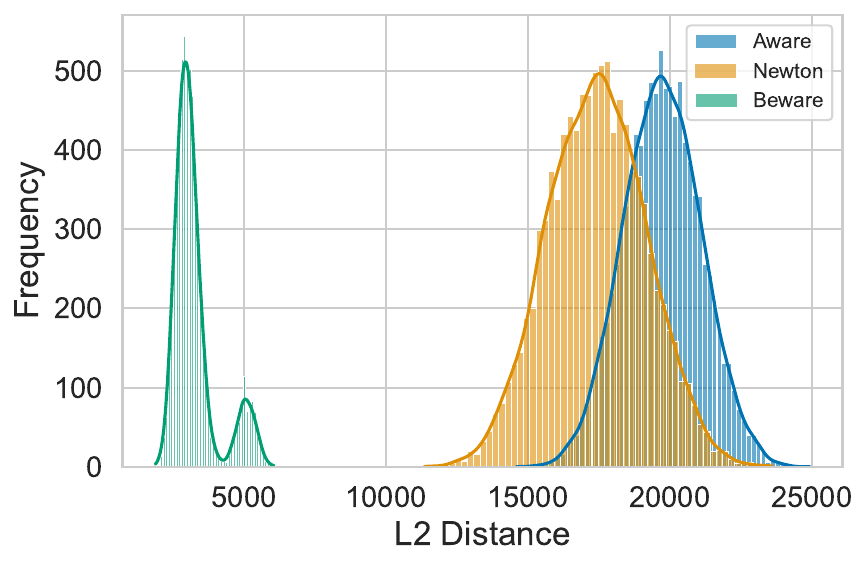}
        \caption{L2 distance between poisoned and sanitized matrices} \label{fig:vcs_sanitization_c}
    \end{subfigure}
    \hfill
    \begin{subfigure}[t]{0.49\textwidth}
        \centering
        \includegraphics[width=0.95\linewidth]{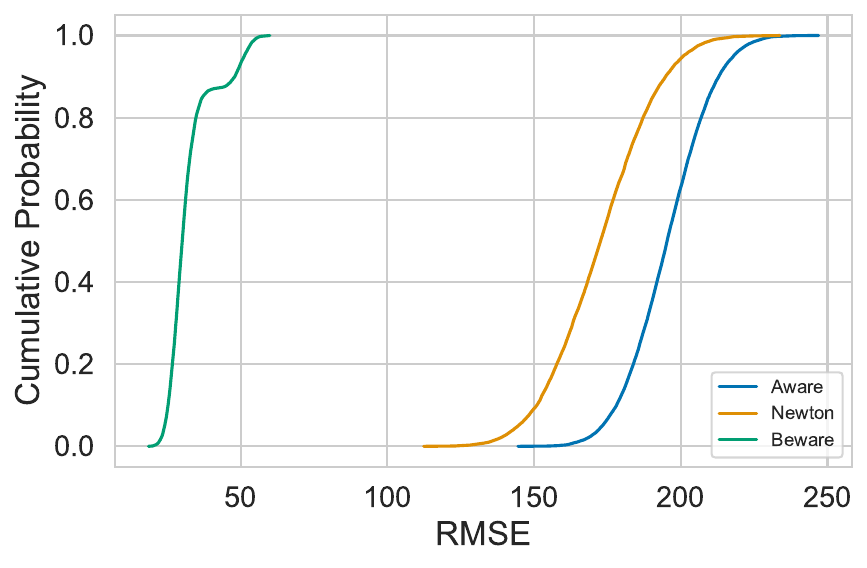}
        \caption{RMSE between sanitized reported and true latencies} \label{fig:vcs_sanitization_d}
    \end{subfigure}
    \hspace*{\fill}
    
    \caption{Accuracy of the sanitization mechanisms of \name{}, Aware, and Newton with inflation attack.
    }
    \label{fig:vcs_sanitization}
\end{figure*}

\subsection{Latency Matrix Sanitization Accuracy}

We consider the inflation attack and evaluate the distance of the sanitized (attacked) matrix using several distance metrics with $f=\Delta=25$, $N=101$, and 10,000 repetitions. The 90th-percentile error highlights the upper bound of sanitization errors, while the number of errors exceeding $100$\,ms captures significant outliers. The L2 distance provides an overall measure of deviation between the sanitized reported and true latency matrices, and the Root Mean Squared Error (RMSE) measures the average magnitude of deviations, with greater sensitivity to larger errors. 

Following the trend of the attack experiment,~\name{}'s sanitization approach filters out malicious latencies more effectively than Aware and Newton. Interestingly, Newton makes more errors above $100$\,ms (Fig.~\ref{fig:vcs_sanitization_a}), whilst Aware makes more significant errors, as indicated by the 90th percentile error rate (Fig.~\ref{fig:vcs_sanitization_b}). 
Finally, the Kernel Density Error plot for L2 distance (Fig.~\ref{fig:vcs_sanitization_c}) and the RMSE CDF plot (Fig.~\ref{fig:vcs_sanitization_d}) show that the cluster-based VCS sanitization approach significantly outperforms all other mechanisms, with no overlap in results. The slight bimodality (visible in both L2 and RMSE) stems from \name{}'s cluster selection algorithm: most runs identify a fully-honest cluster (low error), while a small number identify a cluster that includes some faulty forces (higher error); Aware/Newton do not have this discrete switch. Overall, these results highlight \name{}'s effectiveness in filtering adversarial latencies.

\subsection{Weight Reconfiguration}

To determine the most efficient optimization approach, we evaluate the reconfiguration performance of four \name{} variants against Aware. The variants use Simulated Annealing (SA), Differential Evolution (DE), and their ML-enhanced versions (SA+ML, DE+ML). Aware relies on SA. Experiments are run on a network with $3f+1+\Delta$ nodes ($f=10$, $\Delta=f$), measuring the average number of rounds needed to achieve latency improvements of 0--45\%. Each setup is repeated $20$ times. We evaluate the number of reconfiguration rounds required per baseline to reach a target latency improvement (in \%) with $f=10$ message-delaying faulty nodes. These faulty nodes follow the protocol honestly but add artificial delay to every message they send. This models an attack that targets Limitation~2, in which Byzantine nodes intentionally slow consensus. Reconfiguration frameworks cannot detect this attack from the latency matrix alone, since the nodes report correct latencies but then underperform during actual consensus.

Fig.~\ref{fig:rounds_methods} shows that all \name{} variants outperform Aware, requiring sublinear growth in rounds as thresholds increase. 
Each round on the x-axis of Fig.~\ref{fig:rounds_methods} is one optimizer iteration (i.e., one proposed reconfiguration), and each such iteration internally uses $10$ consensus rounds to measure consensus latency.
SA performs worst, e.g., $\approx 500$ rounds for 25\% improvement and $\approx 800$ for 40\%, due to inefficient random swaps. SA+ML reduces rounds drastically (25\% in $\approx 80$ rounds, 84\% faster than SA) but struggles to exceed 35\% improvements. DE is more efficient than SA, achieving up to 80\% savings for thresholds $\le 30\%$, though it plateaus at higher levels (400--600 rounds for 40\%). DE+ML delivers the best overall performance: fewer than 100 rounds for $\le 25\%$ improvements and $\approx 150$ rounds for 40\%, while maintaining efficiency across thresholds. In summary, SA scales poorly; SA+ML is fast at low thresholds but is limited; DE is stronger, but stalls at higher thresholds; and DE+ML consistently outperforms all others, combining speed and robustness.

\begin{figure}[t]
    \centering
    \begin{minipage}{0.49\textwidth}
        \includegraphics[width=\textwidth]{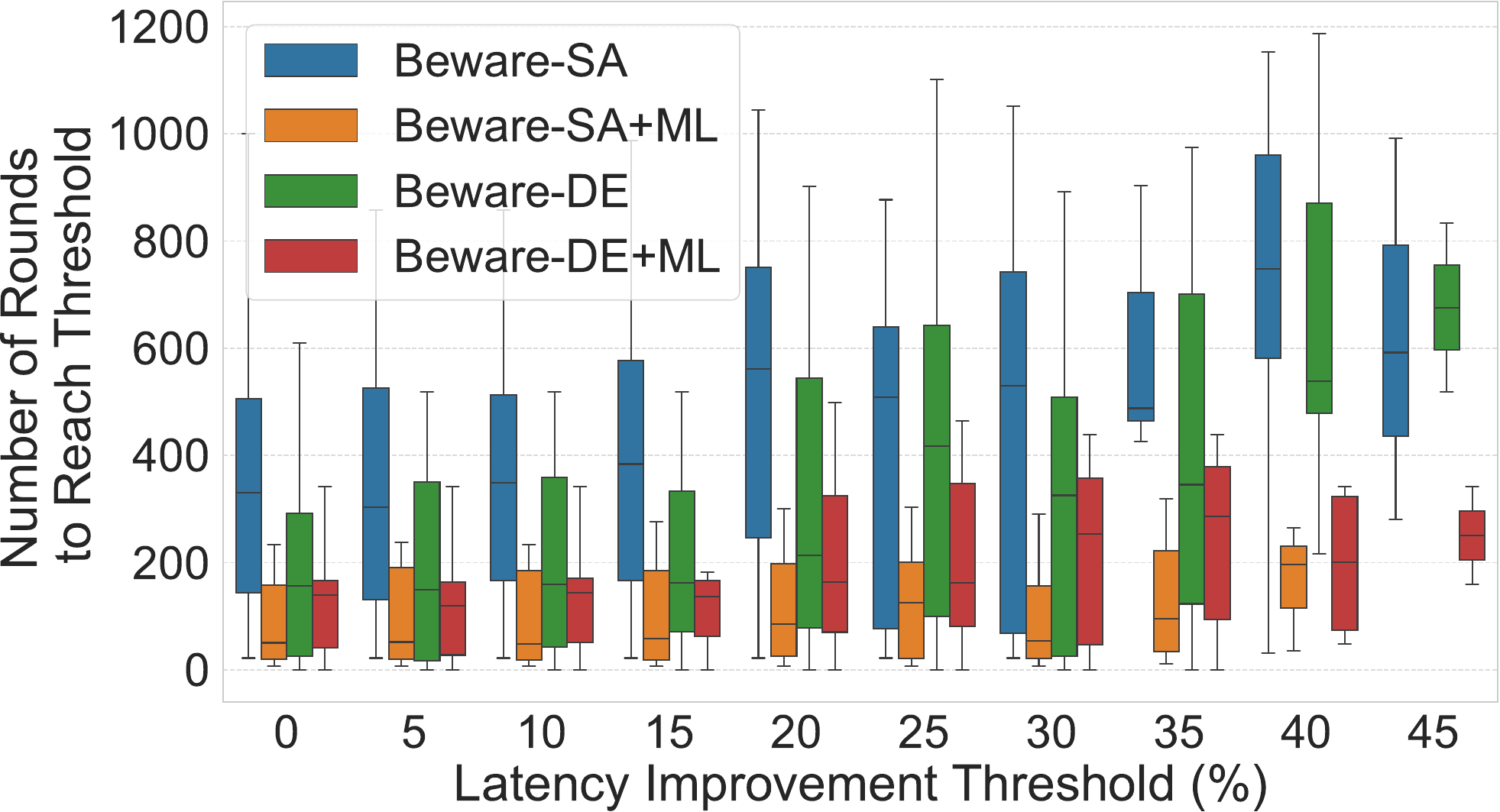}
        \caption{Number of rounds necessary to reach an improvement threshold compared to Aware (\%) with $f{=}10$ message-delaying faulty nodes. Lower is better. }
        \label{fig:rounds_methods}
    \end{minipage}
    \hfill
    \begin{minipage}{0.49\textwidth}
        \includegraphics[width=\linewidth]{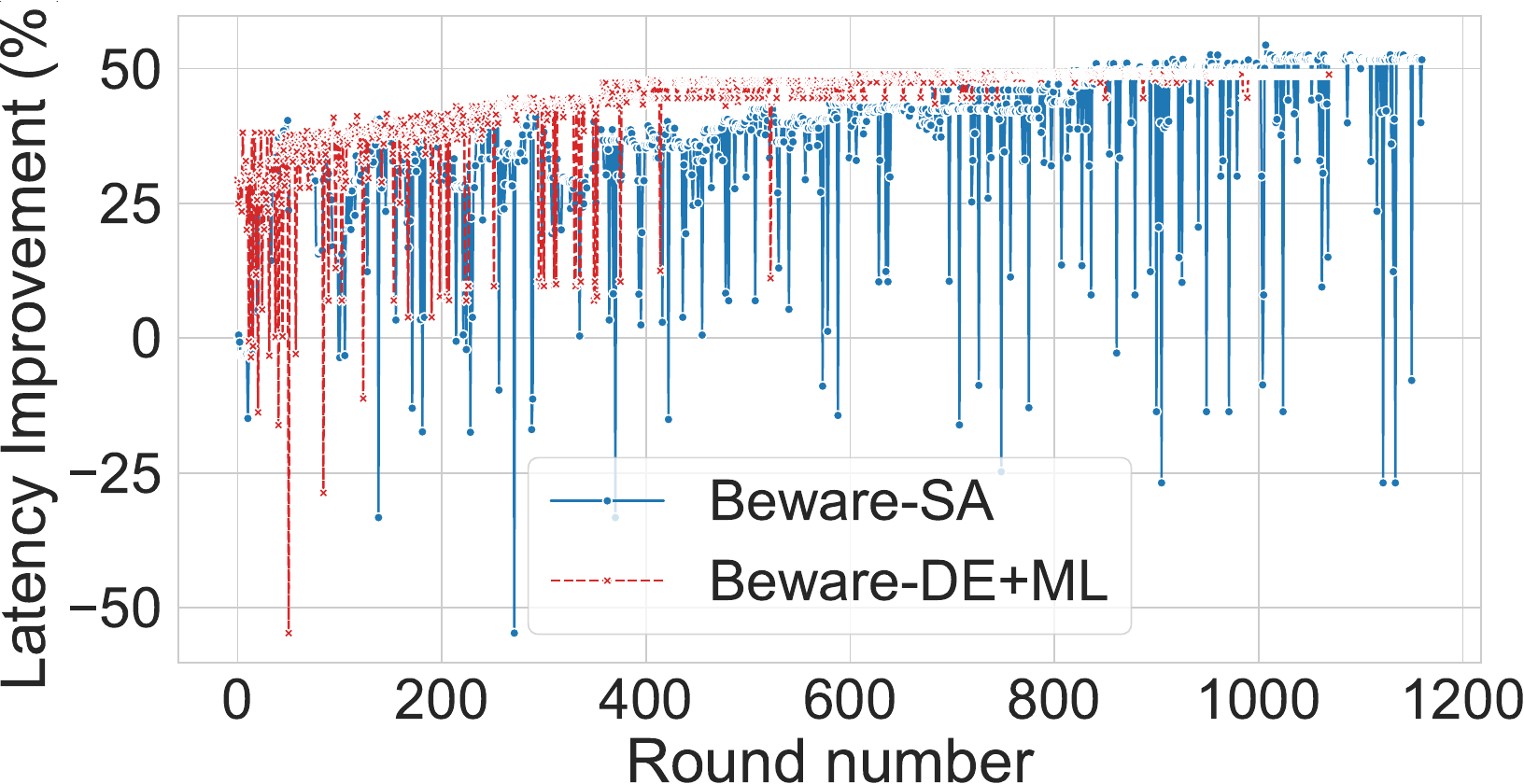}
        \caption{Latency improvement (\%) compared to Aware depending on the round number with $f{=}10$ message-delaying faulty nodes. Higher values are better.}
        \label{fig:fluctuations}
    \end{minipage}
\end{figure}



We further compare SA (baseline) with DE+ML by tracking latency improvements over rounds (Fig.~\ref{fig:fluctuations}) in a network of delaying replicas. SA (blue, solid) shows strong oscillations: even after $1,000$ rounds, improvements dip below $0\%$, and the median never surpasses $+25\%$. In contrast, DE+ML (red, dashed) exceeds $+20\%$ within $60$ rounds, reaches $+35\%$ before round $200$, and then stabilizes with only mild fluctuations.
These results highlight two key points: (i) DE+ML converges much faster, and (ii) it delivers stable improvements, crucial for production deployments. The negative spikes in SA stem from its annealing schedule, in which worse configurations are deliberately accepted to escape local minima, often persisting over multiple consecutive rounds. In DE+ML, the occasional negative dips are artifacts of single-round evaluation and do not indicate the adoption of a worse configuration.

\subsection{Adaptation to network changes}
\label{sec:adaptation_network_changes}




\begin{figure}{t}
        \centering
        \includegraphics[width=\linewidth]{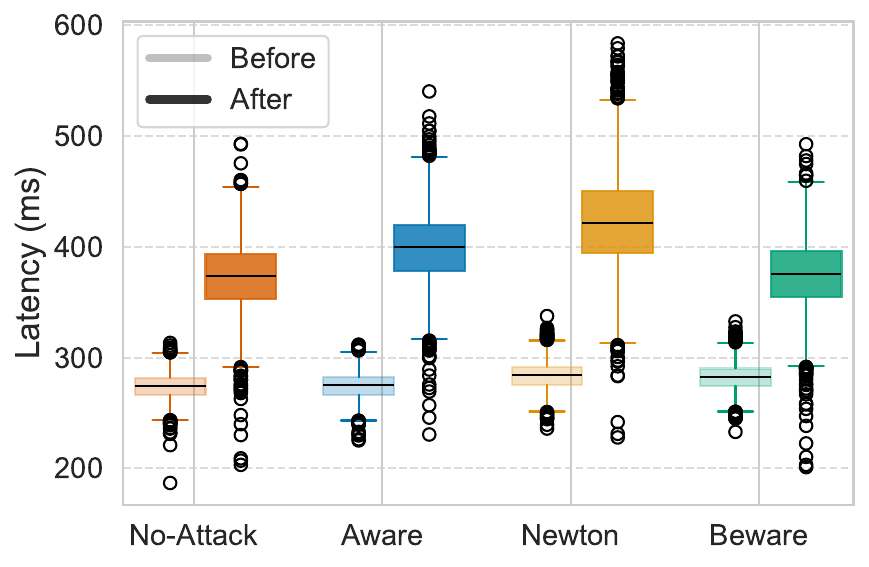}
        \caption{Consensus latency of all baselines before and after reconfiguration under inflation attack.}
        \label{fig:vcs_attack_reconfigure}
\end{figure}

We evaluate \name{} under dynamic conditions using $N{=}101$ and $f{=}25$. Starting from an initial topology, we trigger reconfiguration by introducing inter-cluster congestion via random latency forces.
We also test resilience against adversarial inflation attacks (Fig.~\ref{fig:vcs_attack_reconfigure}). Before reconfiguration, VCS-based methods introduce minor noise. After reconfiguration, however, \name{} closely tracks the optimal baseline, filtering adversarial latencies while adapting to genuine network changes. Competing approaches fail to match this robustness.
Our experiments on synthetic datasets and the WonderNetwork dataset validated \name{}’s performance with its current parameters. We believe its performance can be further improved through hyperparameter tuning on specific WAN latency datasets. 


\subsection{Deployment Experiment}
\label{sec:deployment}

To validate \name{} under realistic network conditions, we deployed it on the DAS academic cluster~\cite{bal2016medium} with $n = 21$ replicas and $f = 5$ Byzantine faults.
As mentioned, computing nodes are equipped with dual 8-core processors
 running at 2.4 GHz, have 60\,GB of RAM, and run the Rocky Linux operating system. Communication between nodes is facilitated
by Gigabit Ethernet (GbE) and InfiniBand (IB), but we emulate network latency based on the WonderNetwork dataset~\cite{wondernetwork} to emulate a wide-area network. Starting from a stable configuration, we provoke two scenarios that follow the trends observed in our simulations. The experiment lasts for 360 seconds.

First, we injected network congestion at $t=90\,s$ for 90 seconds, following the Pareto principle: $20\%$ of the replicas experienced a threefold increase in latency. This models localized network degradation in which only a small subset suffers heavy congestion. As shown in Fig.~\ref{fig:congestion_inflation}, consensus latency briefly spikes above $800$\,ms but quickly recovers as \name{} optimizes voting weights for faster replicas, stabilizing performance with minor fluctuations caused by reconfiguration.

Next, we performed an inflation attack at $t=240\,s$ for a 90-second period, during which the $f=5$ Byzantine replicas attempted to appear more reliable by artificially inflating reported latencies from other nodes (a factor-of-3 misreport). This manipulation initially degrades latency (peaking near $560$\,ms), but \name{} again detects the inconsistency between reported and observed latencies, adjusts weights, and restores consensus performance.

A PBFT baseline (not shown) would keep dishonest replicas influential, causing prolonged degradation.  
We expect Aware~\cite{berger2020aware} to show a similar initial latency spike under congestion and some sensitivity to inflated reports, but---like \name{}---to reweight toward honest replicas and recover. However, as our simulations indicate, \name{} achieves more effective configurations under these conditions, a finding corroborated by our deployment results.   
Overall, these results show that dynamic weighting enables \name{} to sustain low latency under both congestion and inflation attacks, whereas PBFT remains vulnerable.

\begin{figure}[t]
    \centering
    \includegraphics[width=\linewidth]{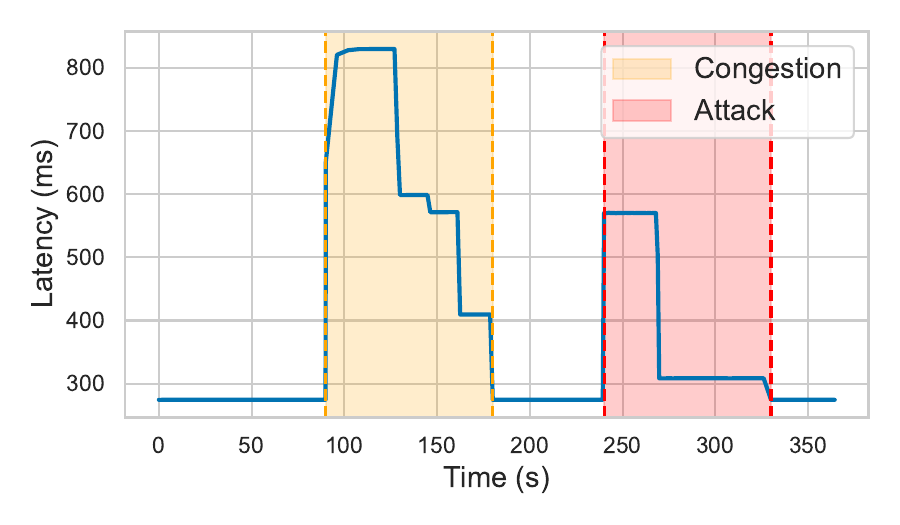}
    \caption{Beware's consensus latency under congestion and inflation attack.}
    \label{fig:congestion_inflation}
\end{figure}

\section{Related Work}
\rc{Add more related work}
\label{sec:sota}


\textbf{Leader Selection in BFT SMR.}
Prior work on BFT State-Machine Replication (BFT SMR) has focused primarily on reducing consensus latency and/or increasing throughput~\cite{lamport2005generalized,sutra2011fast,moraru2013there,du2014clock,liu2016leader,berger2020aware,neiheiser2021kauri,berger2024chasing}, often assuming trusted latency measurements from clients, third parties, or sysadmins. 
Several geo-replicated BFT-SMR systems reduce WAN latency/overhead through fast paths and communication reduction (e.g., SBFT~
\cite{gueta2019sbft}) or by removing the single-leader bottleneck via parallel leaders (e.g., Mir-BFT~\cite{stathakopoulou2022mir}, FnF-BFT~\cite{avarikioti2020fnf}), while bounded-delay protocols such as BFT-Mencius~\cite{milosevic2013bounded} target latency guarantees under stabilizing network conditions.
For instance, ARCHER~\cite{eischer2018latency} selects leaders using client-measured end-to-end response times (typically optimizing a target percentile), obtained via probe requests that run the full BFT protocol and complete when the client receives $f\!+\!1$ matching replies. It mitigates manipulation by requiring probes to be verifiable (preventing replicas from replying early by skipping steps) and by limiting the influence of faulty clients through robust aggregation of client reports. In contrast, Prime~\cite{amir2010prime}, Aardvark~\cite{aardvark}, and RBFT~\cite{aublin2013rbft} rely on replica-side measurements of primary connectivity.

In this work, we assumed that clients are uniformly distributed across the network and focused on reducing consensus latency. To minimize client end-to-end latency with more general client latency distributions, future work could integrate ARCHER's approach into our methods. The end-to-end latency of a configuration would then be the average time until clients receive $f+1$ replies from replicas. Unlike systems that primarily optimize client request ordering, we focus on replica-to-replica communication, which drives consensus latency and enables robust optimizations even when malicious replicas attempt to create bottlenecks. Kauri~\cite{neiheiser2021kauri} instead pipelines requests to lower latency, and other systems~\cite{liu2016leader,eischer2018latency} pursue similar optimizations but remain vulnerable to the attacks we study. Accountability-focused work~\cite{sheng2021bft,civit2023easy} detects faulty replicas through conflicting signed messages, but fails to capture those that delay consensus. In contrast, our system predicts the impact of faulty replicas on latency, limiting their influence to preserve performance. Related work has also shown that Byzantine replicas can satisfy safety and liveness while severely degrading performance via targeted slowdowns and protocol abuse (e.g., slow primary attacks~\cite{spinning}), motivating defenses that explicitly address performance under adversarial behavior~\cite{amir2008byzantine}.
While we apply our framework to PBFT, it could also be used with other consensus algorithms, e.g., to optimize the leader rotation scheme in HotStuff~\cite{yin2019hotstuff}. \\

\textbf{Membership and Protocol Reconfiguration.} Early work like Rampart~\cite{reiter1996distributing,reiter2002secure} established that BFT systems can manage their membership by having their members vote on it. More recently, ComChain~\cite{vizier2020comchain} implemented reconfiguration for permissioned blockchains on top of  DBFT~\cite{Crain+Gramoli+Larrea+Raynal:nca:2018}.  BMS~\cite{steinhoff2021bms} introduced a decentralized membership service applicable to both permissioned and permissionless environments, which active replicas can update via consensus and which clients can contact to identify the latest membership. 
However, Saltini~\cite{saltini2022bigfoot} demonstrated that improper reconfiguration timing can compromise liveness. Their protocol, BigFooT, mitigates this by strictly constraining adoptable memberships. In contrast, \name{} prioritizes optimizing role distribution and voting weights over changes to active membership. Incorporating a secure membership adaptation scheme that supports weighted voting into \name{}, for example, following Heydari et al.'s approach~\cite{heydari2023hard}, is future work. Such an extension would further amplify performance gains by effectively integrating additional, geographically diverse replicas. 
Similarly, \name{} might be extended to dynamically select a different consensus algorithm to execute~\cite{aublin2015next,bahsoun2015making,wu2024towards,wu2025bftbrain}. \\ 

%
%

\textbf{Weighted Voting Reconfiguration.}
In classical BFT, each node has one vote, with quorums defined as $\lceil \frac{n{+}f{+}1}{2} \rceil$ nodes~\cite{malkhi1998byzantine}. Weighted voting~\cite{gifford1979weighted,garciamolina1985assign,paris1986voting} instead assigns different weights to replicas. Garg and Bridgman~\cite{garg2011weighted} defined the weighted Byzantine agreement problem, showing that consensus is possible with more than $n/3$ failures, provided the total weight of the failures is below $1/3$. WHEAT~\cite{sousa2015separating} extends BFT-SMR with extra replicas and weighted quorums, but restricts weights to two values used at fixed frequencies. Aware~\cite{berger2020aware} optimizes latency by dynamically assigning WHEAT weights, with follow-ups for HotStuff~\cite{nischwitz2022raising} and workload-adaptive systems~\cite{kostler2023fluidity}. 

Developed concurrently with our work, Cabinet~\cite{zhang2025cabinet} and OptiLog~\cite{gogada2026optilog} also employ dynamic voting weights to accelerate consensus. 

Cabinet~\cite{zhang2025cabinet} is designed for crash-tolerant environments, while our work addresses the more complex challenges of Byzantine faults. In Cabinet, the reweighting process is centralized: the leader unilaterally computes and distributes weights based on a geometric sequence. In contrast, our reconfiguration process is hardened against manipulation and supports a broader range of voting weights.
Our theoretical framework for weighted dissemination quorums (see Sec.~\ref{sec:theory}) can be adapted to crash-tolerant systems of $2f + 1 + \Delta$ nodes, provided the weights and the quorum weight $Q_T$ satisfy 
\(
    \frac{1}{2} \sum_{i=1}^{n} w_i < Q_T \le \sum_{i=1}^{n} w_i - \sum_{i=1}^{f} w_i.
\) 
Notably, Cabinet specifically targets systems of $2f + 1$ nodes. In such a deployment, the simple majority quorum weight ($\sum w_i / 2$) utilized by Cabinet is correct. However, in larger systems where $\Delta > 0$, which we target for higher latency improvements, the threshold used must be strictly larger than $\sum w_i / 2$ for safety, as indicated by our formulas.

Like \name{}, OptiLog~\cite{gogada2026optilog} is designed for Byzantine environments and shares several of its objectives. However, it focuses on assigning specific roles (e.g., root or internal nodes in a tree-based protocol like Kauri~\cite{neiheiser2021kauri}) to replicas. OptiLog uses Aware's latency sanitization matrix, which we have shown to be vulnerable to latency poisoning attacks, while \name{} relies on a VCS-based latency matrix sanitization. \name{} replicates the reconfiguration process, whereas OptiLog relies on a leader-based reconfiguration and provides accountability through a global log that must be replayed for verification. Nodes that do not behave as expected, based on their communication latency, are eventually eliminated in \name{} thanks to its use of machine learning. Instead, OptiLog uses a suspicion mechanism.    

Extending \name{} to support throughput-oriented optimization is an interesting future direction. This would require the system to be bandwidth-aware. However, throughput cannot be sanitized as latency can. We envision developing specialized distributed auditing mechanisms to produce reliable throughput matrices. Once sanitized, these metrics could be integrated into a multi-criteria objective function, for instance, by maximizing a linear combination of throughput and inverse latency. Dynamic reweighting or stake-based weights have also been studied in permissionless Byzantine settings, e.g., by FlashConsensus~\cite{berger2024chasing}, Heydari et al.~\cite{heydari2023hard}, and Swiper~\cite{tonkikh2024swiper}, and it might be interesting to extend \name{} accordingly. \\

\textbf{Robust Virtual Coordinate Systems.}
Virtual Coordinate Systems (VCS) assign nodes virtual coordinates such that distances approximate network latencies, often using partial measurements, making them useful in large-scale P2P systems. Vivaldi~\cite{dabek2004vivaldi}, for example, models links as virtual springs and seeks a global equilibrium. These springs resist contraction or extension, respectively, when their end nodes are too close or too far apart, with a force proportional to the difference between their resting length (actual latency) and their current length (predicted latency). However, VCSs are vulnerable to data poisoning, where malicious nodes misreport latencies to launch inflation, deflation, or oscillation attacks~\cite{kaafar2006real,beckery2011applying,becker2011securing}. To improve robustness, systems such as Veracity~\cite{sherr2009veracity} validate coordinate updates via voting, whereas Newton~\cite{seibert2013newton} applies physics-inspired invariants to counter Byzantine manipulation. These defenses, however, are tailored to peer-to-peer environments and assume that faults occur only after a benign operation. We instead consider a stronger adversary that can misreport latencies from the very start.

\section{Conclusion}
\label{sec:conclusion}

This paper addresses key challenges in optimizing Byzantine fault-tolerant state-machine replication (BFT-SMR) systems to reduce consensus latency while maintaining resilience against adversaries. Building on existing reconfiguration frameworks, we propose a robust approach that leverages virtual coordinate systems, extended weighted voting schemes, and machine learning models to predict and adopt configurations that minimize latency. Our contributions include identifying sufficient conditions for Byzantine-weighted dissemination quorums, introducing advanced methods for handling Byzantine attacks on latency matrices, and providing an efficient framework that learns from past performance to identify configurations that would achieve the highest performance despite Byzantine nodes attempting to degrade it.
Through extensive simulations across diverse system sizes, network conditions, and attacks, we demonstrated that our methods significantly improve performance and resilience compared to existing solutions. By integrating a predictive model into the optimization loop, we further enhanced adaptability to network changes and mitigated the impact of adversarial behaviors. 
Future work may explore additional machine learning techniques and consider automatically reconfiguring alternative consensus algorithms, such as DAG-based algorithms. 




\bibliographystyle{IEEEtran}
\bibliography{biblio}

@techreport{lamport2005generalized,
	title        = {Generalized consensus and Paxos},
	author       = {Lamport, Leslie},
	year         = 2005,
    institution  = {Microsoft}
}

@inproceedings{garg2011weighted,
  author       = {Vijay K. Garg and
                  John Bridgman},
  title        = {The Weighted Byzantine Agreement Problem},
  booktitle    = {IPDPS},
  year         = {2011},
  doi          = {10.1109/IPDPS.2011.57}
}

@article{saltini2022bigfoot,
  title={{BigFooT}: A robust optimal-latency {BFT} blockchain consensus protocol with dynamic validator membership},
  author={Saltini, Roberto},
  journal={Computer Networks},
  volume={204},
  pages={108632},
  year={2022},
  publisher={Elsevier}
}

@article{reiter2002secure,
  title={A secure group membership protocol},
  author={Reiter, Michael K},
  journal={IEEE TSE},
  volume={22},
  number={1},
  pages={31--42},
  year={2002},
  publisher={IEEE}
}

@article{reiter1996distributing,
  title={Distributing trust with the Rampart toolkit},
  author={Reiter, Michael},
  journal={Communications of the ACM},
  volume={39},
  number={4},
  pages={71--74},
  year={1996},
  publisher={ACM New York, NY, USA}
}

@article{vizier2020comchain,
  title={{ComChain}: A blockchain with {Byzantine} fault-tolerant reconfiguration},
  author={Vizier, Guillaume and Gramoli, Vincent},
  journal={Concurrency and Computation: Practice and Experience},
  volume={32},
  number={12},
  pages={e5494},
  year={2020},
  publisher={Wiley Online Library}
}

@article{steinhoff2021bms,
  title={{BMS}: Secure decentralized reconfiguration for blockchain and {BFT} systems},
  author={Steinhoff, Selma and Stathakopoulou, Chrysoula and Pavlovic, Matej and Vukoli{\'c}, Marko},
  journal={arXiv preprint arXiv:2109.03913},
  year={2021}
}

@article{storn1997differential,
  title={Differential evolution--a simple and efficient heuristic for global optimization over continuous spaces},
  author={Storn, Rainer and Price, Kenneth},
  journal={Journal of global optimization},
  volume={11},
  pages={341--359},
  year={1997},
  publisher={Springer}
}

@inproceedings{xgboost,
	title        = {{XGBoost}: A Scalable Tree Boosting System},
	author       = {Chen, Tianqi and Guestrin, Carlos},
	year         = 2016,
	booktitle    = {SIGKDD}
}

@inproceedings{paris1986voting,
	title        = {Voting with Witnesses: {A} Constistency Scheme for Replicated Files},
	author       = {Jehan{-}Fran{\c{c}}ois P{\^{a}}ris},
	year         = 1986,
	booktitle    = {ICDCS},
}

@inproceedings{becker2011securing,
	title        = {Securing application-level topology estimation networks: Facing the frog-boiling attack},
	author       = {Becker, Sheila and Seibert, Jeff and Nita-Rotaru, Cristina and State, Radu},
	year         = 2011,
	booktitle    = {RAID}
}

@article{lamport2010reconfiguring,
	title        = {Reconfiguring a state machine},
	author       = {Lamport, Leslie and Malkhi, Dahlia and Zhou, Lidong},
	year         = 2010,
	journal      = {ACM SIGACT News},
	publisher    = {ACM New York, NY, USA},
	volume       = 41,
	number       = 1,
	pages        = {63--73}
}

@inproceedings{beckery2011applying,
	title        = {Applying game theory to analyze attacks and defenses in virtual coordinate systems},
	author       = {Beckery, Sheila and Seibert, Jeff and Zage, David and Nita-Rotaru, Cristina and State, Radu},
	year         = 2011,
	booktitle    = {DSN}
}

@inproceedings{silva2021threat,
	title        = {Threat adaptive byzantine fault tolerant state-machine replication},
	author       = {Silva, Douglas Simoes and Graczyk, Rafal and Decouchant, J{\'e}r{\'e}mie and V{\"o}lp, Marcus and Esteves-Verissimo, Paulo},
	year         = 2021,
	booktitle    = {SRDS}
}

@inproceedings{castro1999practical,
	title        = {Practical byzantine fault tolerance},
	author       = {Castro, Miguel and Liskov, Barbara},
	year         = 1999,
	booktitle    = {OSDI}
}

@article{zhang2025cabinet,
  title={Cabinet: Dynamically Weighted Consensus Made Fast},
  author={Zhang, Gengrui and Zhang, Shiquan and Bachras, Michail and Zhang, Yuqiu and Jacobsen, Hans-Arno},
  journal={VLDB},
  volume={18},
  number={5},
  pages={1439--1452},
  year={2025},
  publisher={VLDB Endowment}
}

@inproceedings{sheng2021bft,
	title        = {{BFT} protocol forensics},
	author       = {Sheng, Peiyao and Wang, Gerui and Nayak, Kartik and Kannan, Sreeram and Viswanath, Pramod},
	year         = 2021,
	booktitle    = {CCS}
}

@article{civit2023easy,
	title        = {As easy as {ABC}: Optimal (A)ccountable (B)yzantine (C)onsensus is easy!},
	author       = {Civit, Pierre and Gilbert, Seth and Gramoli, Vincent and Guerraoui, Rachid and Komatovic, Jovan},
	year         = 2023,
	journal      = {JPDC},
	volume       = 181,
	pages        = 104743
}

@inproceedings{sutra2011fast,
	title        = {Fast genuine generalized consensus},
	author       = {Sutra, Pierre and Shapiro, Marc},
	year         = 2011,
	booktitle    = {SRDS}
}

@inproceedings{moraru2013there,
	title        = {There is more consensus in egalitarian parliaments},
	author       = {Moraru, Iulian and Andersen, David G and Kaminsky, Michael},
	year         = 2013,
	booktitle    = {SOSP}
}

@inproceedings{yin2019hotstuff,
	title        = {{HotStuff}: {BFT} consensus with linearity and responsiveness},
	author       = {Yin, Maofan and Malkhi, Dahlia and Reiter, Michael K and Gueta, Guy Golan and Abraham, Ittai},
	year         = 2019,
	booktitle    = {PODC}
}

@inproceedings{tonkikh2024swiper,
  title={Swiper: a new paradigm for efficient weighted distributed protocols},
  author={Tonkikh, Andrei and Freitas, Luciano},
  booktitle={PODC},
  year={2024}
}

@misc{wondernetwork,
  title        = {WonderNetwork},
  howpublished = {https://wondernetwork.com/},
  author       = {Reinheimer, Paul and Roberts, Will}
}

@inproceedings{bessani2014state,
	title        = {State machine replication for the masses with {BFT-SMART}},
	author       = {Bessani, Alysson and Sousa, Jo{\~a}o and Alchieri, Eduardo},
	year         = 2014,
	booktitle    = {DSN}
}

@inproceedings{heydari2023hard,
  title={How hard is asynchronous weight reassignment?},
  author={Heydari, Hasan and Silvestre, Guthemberg and Bessani, Alysson},
  booktitle={ICDCS},
  year={2023}
}

@article{garciamolina1985assign,
	title        = {How to Assign Votes in a Distributed System},
	author       = {Hector Garcia{-}Molina and Daniel Barbar{\'{a}}},
	year         = 1985,
	journal      = {J. {ACM}},
	volume       = 32,
	number       = 4,
	pages        = {841--860},
	doi          = {10.1145/4221.4223}
}

@inproceedings{du2014clock,
	title        = {Clock-RSM: Low-latency inter-datacenter state machine replication using loosely synchronized physical clocks},
	author       = {Du, Jiaqing and Sciascia, Daniele and Elnikety, Sameh and Zwaenepoel, Willy and Pedone, Fernando},
	year         = 2014,
	booktitle    = {DSN}
}

@article{liu2016leader,
	title        = {Leader set selection for low-latency geo-replicated state machine},
	author       = {Liu, Shengyun and Vukoli{\'c}, Marko},
	year         = 2016,
	journal      = {IEEE TPDS},
	publisher    = {IEEE},
	volume       = 28,
	number       = 7,
	pages        = {1933--1946}
}

@inproceedings{eischer2018latency,
	title        = {Latency-aware leader selection for geo-replicated Byzantine fault-tolerant systems},
	author       = {Eischer, Michael and Distler, Tobias},
	year         = 2018,
	booktitle    = {DSN-W}
}

@inproceedings{neiheiser2021kauri,
	title        = {Kauri: Scalable bft consensus with pipelined tree-based dissemination and aggregation},
	author       = {Neiheiser, Ray and Matos, Miguel and Rodrigues, Lu{\'\i}s},
	year         = 2021,
	booktitle    = {SOSP}
}

@article{dabek2004vivaldi,
	title        = {Vivaldi: A decentralized network coordinate system},
	author       = {Dabek, Frank and Cox, Russ and Kaashoek, Frans and Morris, Robert},
	year         = 2004,
	journal      = {CCR},
	volume       = 34,
	number       = 4,
	pages        = {15--26},
	organization = {ACM SIGCOMM}
}

@article{seibert2013newton,
	title        = {Newton: securing virtual coordinates by enforcing physical laws},
	author       = {Seibert, Jeff and Becker, Sheila and Nita-Rotaru, Cristina and State, Radu},
	year         = 2013,
	journal      = {ToN},
	publisher    = {IEEE/ACM},
	volume       = 22,
	number       = 3,
	pages        = {798--811}
}

@inproceedings{berger2024chasing,
	title        = {Chasing Lightspeed Consensus: Fast Wide-Area Byzantine Replication with Mercury},
	author       = {Berger, Christian and Rodrigues, L{\'\i}vio and Reiser, Hans P and Cogo, Vinicius and Bessani, Alysson},
	year         = 2024,
	booktitle    = {Middleware}
}

@inproceedings{gifford1979weighted,
	title        = {Weighted voting for replicated data},
	author       = {Gifford, David K},
	year         = 1979,
	booktitle    = {SOSP}
}

@inproceedings{nischwitz2022raising,
	title        = {Raising the AWAREness of BFT Protocols for Soaring Network Delays},
	author       = {Nischwitz, Martin and Esche, Marko and Tschorsch, Florian},
	year         = 2022,
	booktitle    = {LCN}
}

@inproceedings{kostler2023fluidity,
	title        = {Fluidity: Location-Awareness in Replicated State Machines},
	author       = {K{\"o}stler, Johannes and Reiser, Hans P and Hauck, Franz J and Habiger, Gerhard},
	year         = 2023,
	booktitle    = {SAC}
}

@inproceedings{sherr2009veracity,
	title        = {Veracity: Practical Secure Network Coordinates via Vote-based Agreements},
	author       = {Sherr, Micah and Blaze, Matt and Loo, Boon Thau},
	year         = 2009,
	booktitle    = {Usenix ATC}
}

@inproceedings{aardvark,
	title        = {Making Byzantine fault tolerant systems tolerate Byzantine faults},
	author       = {Clement, Allen and Wong, Edmund and Alvisi, Lorenzo and Dahlin, Mike and Marchetti, Mirco and others},
	year         = 2009,
	booktitle    = {NSDI}
}

@inproceedings{aublin2013rbft,
	title        = {{RBFT}: Redundant byzantine fault tolerance},
	author       = {Aublin, Pierre-Louis and Mokhtar, Sonia Ben and Qu{\'e}ma, Vivien},
	year         = 2013,
	booktitle    = {ICDCS}
}

@article{amir2010prime,
	title        = {Prime: Byzantine replication under attack},
	author       = {Amir, Yair and Coan, Brian and Kirsch, Jonathan and Lane, John},
	year         = 2010,
	journal      = {TDSC},
	publisher    = {IEEE},
	volume       = 8,
	number       = 4,
	pages        = {564--577}
}

@inproceedings{bahsoun2015making,
	title        = {Making BFT protocols really adaptive},
	author       = {Bahsoun, Jean-Paul and Guerraoui, Rachid and Shoker, Ali},
	year         = 2015,
	booktitle    = {IPDPS}
}

@inproceedings{kaafar2006real,
	title        = {Real attacks on virtual networks: Vivaldi out of tune},
	author       = {Kaafar, Mohamed Ali and Mathy, Laurent and Turletti, Thierry and Dabbous, Walid},
	year         = 2006,
	booktitle    = {SIGCOMM LSAD Workshop}
}

@article{berger2020aware,
	title        = {{AWARE}: Adaptive wide-area replication for fast and resilient {Byzantine} consensus},
	author       = {Berger, Christian and Reiser, Hans P and Sousa, Jo{\~a}o and Bessani, Alysson Neves},
	year         = 2020,
	journal      = {IEEE TDSC},
	publisher    = {IEEE}
}

@inproceedings{sousa2015separating,
	title        = {Separating the WHEAT from the chaff: An empirical design for geo-replicated state machines},
	author       = {Sousa, Jo{\~a}o and Bessani, Alysson},
	year         = 2015,
	booktitle    = {SRDS}
}

@inproceedings{spinning,
	title        = {Spin One's Wheels? {Byzantine} Fault Tolerance with a Spinning Primary},
	author       = {Veronese, Giuliana Santos and Correia, Miguel and Bessani, Alysson Neves and Lung, Lau Cheuk},
	year         = 2009,
	booktitle    = {SRDS}
}

@inproceedings{cox2025catalyst,
  title={Catalyst: Asynchronous Byzantine-Robust Federated Learning},
  author={Cox, Bart and M{\u{a}}lan, Abele and Chen, Lydia Y and Decouchant, J{\'e}r{\'e}mie},
  booktitle={BigData},
  year={2025}
}

@inproceedings{gogada2026optilog,
  title={{OptiLog}: Assigning Roles in Byzantine Consensus},
  booktitle={EuroSys},
  author={Gogada, Hanish and Berger, Christian and Jehl, Leander and Reiser, Hans P and Meling, Hein},
  year={2026}
}

@inproceedings{Castro+Liskov:osdi:pbft:1999,
	title        = {Practical Byzantine Fault Tolerance},
	author       = {Miguel Castro and Barbara Liskov},
	year         = 1999,
	booktitle    = {OSDI}
}

@inproceedings{campello2013density,
	title        = {Density-Based Clustering Based on Hierarchical Density Estimates},
	author       = {Ricardo J. G. B. Campello and Davoud Moulavi and J{\"{o}}rg Sander},
	year         = 2013,
	booktitle    = {{PAKDD}}
}

@inproceedings{nguyen2022flame,
	title        = {{FLAME}: Taming backdoors in federated learning},
	author       = {Nguyen, Thien Duc and Rieger, Phillip and De Viti, Roberta and Chen, Huili and Brandenburg, Bj{\"o}rn B and Yalame, Hossein and M{\"o}llering, Helen and Fereidooni, Hossein and Marchal, Samuel and Miettinen, Markus and others},
	year         = 2022,
	booktitle    = {USENIX Security}
}

@article{fereidooni2023freqfed,
	title        = {FreqFed: A Frequency Analysis-Based Approach for Mitigating Poisoning Attacks in Federated Learning},
	author       = {Fereidooni, Hossein and Pegoraro, Alessandro and Rieger, Phillip and Dmitrienko, Alexandra and Sadeghi, Ahmad-Reza},
	year         = 2023,
	journal      = {arXiv preprint arXiv:2312.04432}
}

@inproceedings{Crain+Gramoli+Larrea+Raynal:nca:2018,
	title        = {{DBFT:} Efficient Leaderless Byzantine Consensus and its Application to Blockchains},
	author       = {Tyler Crain and Vincent Gramoli and Mikel Larrea and Michel Raynal},
	year         = 2018,
  	booktitle    = {NCA},
	doi          = {10.1109/NCA.2018.8548057},
	bibsource    = {dblp computer science bibliography, https://dblp.org}
}

@inproceedings{civit2021polygraph,
	title        = {Polygraph: Accountable byzantine agreement},
	author       = {Civit, Pierre and Gilbert, Seth and Gramoli, Vincent},
	year         = 2021,
	booktitle    = {ICDCS}
}

@article{malkhi1998byzantine,
	title        = {Byzantine quorum systems},
	author       = {Malkhi, Dahlia and Reiter, Michael},
	year         = 1998,
	journal      = {Distributed computing},
	publisher    = {Springer},
	volume       = 11,
	number       = 4,
	pages        = {203--213}
}

@article{wu2024towards,
  title={Towards truly adaptive byzantine fault-tolerant consensus},
  author={Wu, Chenyuan and Qin, Haoyun and Javad Amiri, Mohammad and Thau Loo, Boon and Malkhi, Dahlia and Marcus, Ryan},
  journal={ACM SIGOPS Operating Systems Review},
  volume={58},
  number={1},
  pages={15--22},
  year={2024},
  publisher={ACM New York, NY, USA}
}

@inproceedings{wu2025bftbrain,
  title={{BFTBrain}: Adaptive {BFT} consensus with reinforcement learning},
  author={Wu, Chenyuan and Qin, Haoyun and Amiri, Mohammad Javad and Loo, Boon Thau and Malkhi, Dahlia and Marcus, Ryan},
  booktitle={NSDI},
  year={2025}
}

@article{aublin2015next,
  title={The next 700 BFT protocols},
  author={Aublin, Pierre-Louis and Guerraoui, Rachid and Kne{\v{z}}evi{\'c}, Nikola and Qu{\'e}ma, Vivien and Vukoli{\'c}, Marko},
  journal={ACM Transactions on Computer Systems (TOCS)},
  volume={32},
  number={4},
  pages={1--45},
  year={2015},
  publisher={ACM New York, NY, USA}
}

@article{de2022noise,
  title={Noise in the clouds: Influence of network performance variability on application scalability},
  author={De Sensi, Daniele and De Matteis, Tiziano and Taranov, Konstantin and Di Girolamo, Salvatore and Rahn, Tobias and Hoefler, Torsten},
  journal={POMACS},
  volume={6},
  number={3},
  pages={1--27},
  year={2022},
  publisher={ACM New York, NY, USA}
}

@inproceedings{uta2018performance,
  title={A performance study of big data workloads in cloud datacenters with network variability},
  author={Uta, Alexandru and Obaseki, Harry},
  booktitle={ICPE Companion},
  pages={113--118},
  year={2018}
}

@inproceedings{iosup2011performance,
  title={On the performance variability of production cloud services},
  author={Iosup, Alexandru and Yigitbasi, Nezih and Epema, Dick},
  booktitle={CCGrid},
  pages={104--113},
  year={2011},
  organization={IEEE}
}

@article{bal2016medium,
  title={A medium-scale distributed system for computer science research: Infrastructure for the long term},
  author={Bal, Henri and Epema, Dick and De Laat, Cees and Van Nieuwpoort, Rob and Romein, John and Seinstra, Frank and Snoek, Cees and Wijshoff, Harry},
  journal={Computer},
  volume={49},
  number={5},
  pages={54--63},
  year={2016},
  publisher={IEEE}
}

@inproceedings{amir2008byzantine,
  title={Byzantine replication under attack},
  author={Amir, Yair and Coan, Brian and Kirsch, Jonathan and Lane, John},
  booktitle={DSN},
  year={2008}
}

@inproceedings{gueta2019sbft,
  title={{SBFT}: A scalable and decentralized trust infrastructure},
  author={Gueta, Guy Golan and Abraham, Ittai and Grossman, Shelly and Malkhi, Dahlia and Pinkas, Benny and Reiter, Michael and Seredinschi, Dragos-Adrian and Tamir, Orr and Tomescu, Alin},
  booktitle={DSN},
  year={2019}
}

@article{stathakopoulou2022mir,
  title={Mir-{BFT}: Scalable and robust {BFT} for decentralized networks},
  author={Stathakopoulou, Chrysoula and David, Tudor and Pavlovic, Matej and Vukolic, Marko},
  journal={Journal of Systems Research},
  volume={2},
  number={1},
  pages={3},
  year={2022}
}

@article{avarikioti2020fnf,
  title={Fnf-bft: Exploring performance limits of BFT protocols},
  author={Avarikioti, Zeta and Heimbach, Lioba and Schmid, Roland and Vanbever, Laurent and Wattenhofer, Roger and Wintermeyer, Patrick},
  journal={arXiv preprint arXiv:2009.02235},
  year={2020}
}

@inproceedings{milosevic2013bounded,
  title={Bounded delay in byzantine-tolerant state machine replication},
  author={Milosevic, Zarko and Biely, Martin and Schiper, Andr{\'e}},
  booktitle={SRDS},
  year={2013}
}


\end{document}